\documentclass[12pt]{iopart}
\usepackage[utf8]{inputenc}

\usepackage[english]{babel}
\usepackage{mathrsfs}
\expandafter\let\csname equation*\endcsname\relax
\expandafter\let\csname endequation*\endcsname\relax
\usepackage{amsmath}
\usepackage{amsthm}
\usepackage{amsfonts}
\usepackage{BOONDOX-cal}
\usepackage{graphicx}
\usepackage{braket}
\usepackage{subcaption}
\usepackage{multirow}
\usepackage[table]{xcolor}
\usepackage{MnSymbol}
\usepackage{caption}
\usepackage{subcaption}
\usepackage{float}
\usepackage{mathrsfs}
\usepackage{amsthm}
\usepackage{graphicx}
\usepackage{dcolumn}
\usepackage{bm}
\usepackage{perpage}
\usepackage{cite}
\usepackage[mathlines]{lineno} 

\newcommand{\configSpace}{\mathcal{C}}
\newcommand{\suppPsi}{\mathscr{S}}

\newcommand{\innerProduct}[2]{\langle #1, #2 \rangle}
\newcommand{\op}[1]{\hat{#1}}
\newcommand{\farcomma}{\quad,\quad}

\newcommand{\expfunc}[1]{\exp \left[#1\right]}
\newcommand{\expect}[1]{\left\langle #1 \right\rangle}
\newcommand{\newpara}{\\\\\noindent}
\newcommand{\ii}{\mathrm{i}}
\newcommand{\ev}[2]{\mathrm{ev}_{#1}(#2)}

\newtheorem{theorem}{Theorem}[section]
\newtheorem{lemma}{Lemma}[section]
\newtheorem{definition}{Definition}[section]
\newtheorem{corollary}{Corollary}[section]
\newtheorem{remark}{Remark}[section]

\begin{document}

\title[]{Simultaneous approximation of multiple degenerate states using a single neural network quantum state}

\author{Waleed Sherif}

\address{Institute for Quantum Gravity, Department of Physics, Friedrich-Alexander-Universit\"{a}t Erlangen-N\"{u}rnberg (FAU), Staudtstraße 7, 91058 Erlangen, Germany}
\ead{waleed.sherif@gravity.fau.de}
\vspace{10pt}
\begin{indented}
\item[] \today
\end{indented}

\begin{abstract}
Neural network quantum states (NQS) excel at approximating ground states of quantum many-body systems, but approximating all states of a degenerate manifold is nevertheless computationally expensive. We propose a single-trunk multi-head (ST-MH) NQS ensemble that share a feature extracting trunk while attaching lightweight heads for each target state. Using a cost function which also has an orthogonality term, we derive exact analytic gradients and overlap derivatives needed to train ST-MH within standard variational Monte Carlo (VMC) workflows. We prove that ST-MH can represent every degenerate eigenstate exactly whenever the feature map of latent width $h$, augmented with a constant, has column space containing the linear span of the targets' log-moduli and (chosen) phase branches together with the constant on the common support where all states are non-vanishing. Under this condition, ST-MH reduces the parameter count and can reduce the leading VMC cost by a factor equal to the degeneracy $K$ relative to other algorithms when $K$ is modest and in trunk dominated regimes. As a numerical proof-of-principle, we validate and benchmark the ST-MH approach on the frustrated spin-$\tfrac{1}{2}$ $J_1-J_2$ Heisenberg model at the Majumdar-Ghosh point on periodic ring lattices of up to 8 sites. By obtaining the momentum eigenstates, we demonstrate that ST-MH attains high fidelity and energy accuracy across degenerate ground state manifolds while using significantly lower computing resources. Lastly we provide a qualitative computational cost analysis which incentivise the applicability of the ST-MH ensemble under certain criteria on the latent width.
\end{abstract}


\section{\label{sec:introduction}Introduction}

In quantum many-body physics, the ground state of a system holds central significance. It encodes the fundamental properties of the system at zero temperature and determines phase behaviour, correlation structure, and, in combination with excited states, response functions. For strongly interacting systems, exact solutions are typically unavailable and one often relies on stochastic methods. Quantum Monte Carlo (QMC) methods \cite{Ceperley,Foulkes:2001zz,Carlson:2014vla,Becca_Sorella_2017} are regarded as highly accurate and widely used to approximate ground states numerically. Among them, variational Monte Carlo (VMC) \cite{Becca_Sorella_2017} provides a robust framework in which one parametrises the wave-function using some variational ansatz and then optimises its parameters to minimise the expectation value of the Hamiltonian.
\newpara
Among the host of variational Ans\"{a}tze \cite{Orus:2013kga,PhysRevLett.93.207205,PhysRevLett.112.240502,PhysRevLett.99.220405,PhysRevA.74.022320,PhysRevLett.100.240603,Verstraete:2004cf,Klumper:1993oof,SCHOLLWOCK201196,PhysRevLett.69.2863,RevModPhys.77.259}, neural network quantum states (NQS) \cite{Carleo:2017nvk} have emerged as a powerful ansatz for this problem, specifically for strongly correlated quantum systems. They are capable of representing volume-law entangled states \cite{PhysRevB.106.205136,PhysRevX.7.021021,Gao2017,Passetti:2022ilw,PhysRevLett.122.065301} and consequently have been successfully applied to a broad spectrum of quantum systems, from lattice models \cite{Carleo:2017nvk,PhysRevB.100.125124,Luo:2021gls}, fermionic systems \cite{PhysRevResearch.2.033429,Hermann2020}, bosonic field theories \cite{Luo:2020stn,Martyn:2022oll,Denis:2025njk}, where conventional methods often struggle due to the sign problem \cite{PAN2024879,Troyer:2004ge} or limited representational flexibility, and recently even in loop quantum gravity \cite{Sahlmann:2024pba,Sahlmann:2024kat}. Aside from excited states applications, one typical use-case is focused on learning a single eigenstate, namely the non-degenerate ground state of the system. 
\newpara
In the case of degenerate (ground) spaces, the most naive manner to obtain multiple degenerate states is to optimise independent networks for each target state. This is both computationally expensive and susceptible to convergence toward redundant or non-orthogonal solutions, especially when degeneracies arise from symmetries or frustration. This led to the development of several general methods \cite{Higgott:2018doo,10.1063/5.0030949,Pfau:2023azx,PhysRevB.111.L161116,PhysRevLett.121.167204,Wheeler_2024} in the context of VMC to overcome such issues.
\newpara
A state-averaging approach has been used to simultaneously optimise a set of ground and excited states, where a common set of orbitals and Jastrow are used to construct all the states \cite{Filippi2009}. Recently, a general strategy to address this problem has been put forth whereby one uses an ensemble-based learning process \cite{Wheeler_2024}. In it, each of the $K$-many degenerate states is represented by a separate copy of the used variational ansatz, each with its own set of parameters. This setup naturally supports orthogonalisation, as the cost function is modified to include an overlap and orthogonality penalty term \cite{Wheeler_2024}. In the neural network context, this corresponds to having $K$ independent networks. The cost of this $K$-duplications of the network entails a significant increase in the total parameter count. In the context of NQS, since parameter updates occur after computing gradients of the cost function with respect to the parameters, one very quickly is met with a computational bottleneck which may severely limit scalability.
\newpara
Perhaps a bit more distant to the problem of finding ground states, there is a growing class of efficient physics machine learning (ML) algorithms that learn multiple related quantities from a shared representation with light, task-specific outputs. Notably, multi-head physics informed neural networks \cite{Tarancón-Álvarez2025,Pellegrin2022,lhydra} such as L-HYDRA \cite{lhydra} employ a single non-linear ``body" with several linear heads to address families of partial differential equation tasks, and multi-state neural models in quantum chemistry predict several state energies and couplings from a common latent space \cite{Axelrod2022}.
\newpara
In this work, motivated by multi-task learning in classical ML, by these physics analogues, and by the shared-orbitals/Jastrow practice \cite{Filippi2009}, we propose an efficient ensemble approach when utilising NQS for approximating degenerate ground states: a single-trunk multi-head (ST-MH) ensemble, whereby one network is used to approximate $K$-many degenerate states. 
\newpara
This approach retains a single shared feature-extracting trunk while appending lightweight, linearly parametrised heads for each target eigenstate. We demonstrate that this construction is not merely a heuristic compression but that, under precise conditions on the trunk's width, it can represent the entire degenerate manifold exactly. Specifically, if the latent width $h$ of the shared trunk satisfies $h + 1 \geq r_\mathrm{both}$, where $r_\mathrm{both}$ is the combined linear rank of the states' log-moduli and phases on a common support in the degenerate target manifold where the states are non-vanishing (see Appendix \ref{app:sufficiencytheorem}, Definition \ref{def:rank}), then all degenerate eigenstates can be captured without loss of expressivity. Furthermore, we derive closed form analytic gradients for the energy and overlap penalties as well as sampling strategies, enabling integration into standard VMC workflows.
\newpara
We demonstrate the applicability of the proposed NQS ensemble approach by obtaining the degenerate momentum eigenstates of the frustrated spin-$\tfrac{1}{2}$ $J_1-J_2$ Heisenberg chain on a periodic ring with even sites at the Majumdar-Ghosh point \cite{C_K_Majumdar_1970,10.1063/1.1664978,10.1063/1.1664979}. Across various metrics, including fidelity, orthogonality, memory and runtime footprint, we show that the ST-MH NQS ensemble achieves comparable accuracy all the meanwhile maintaining a substantially lower resource demand. We also provide a qualitative analysis of the gained efficiency in compute cost when using the ST-MH NQS ensemble and specify a qualitative threshold which takes into account both $K$ and the trunk width to obtain such efficiency. This offers a qualitative threshold regarding the applicability of this proposed approach based on the number of target states. The conclusion made is therefore that the proposed ensemble approach offers a principled and practical route to learning degenerate eigenspaces efficiently, when optimising using the penalty based cost method, provided that the linear rank condition $r_\mathrm{both} \leq h + 1$ is satisfied and that the phase functions admit single-valued branches that can be represented as affine functions of the shared trunk features (mod $2\pi$).


\section{\label{sec:vmc}NQS Ensembles}

The goal of variational Monte Carlo is to iteratively optimise some ansatz wave-function to approximate a ground state \cite{Becca_Sorella_2017}. In this work, we will focus on approximating several ground states and the chosen variational ansatz is the neural network quantum state ansatz. In the non-degenerate ground space case, VMC optimisation aims to minimise the cost function
\begin{equation}
    \label{eq:cost_functional_nqs}
    C := \frac{\innerProduct{\Psi_\theta}{\op{H}\Psi_\theta}}{\innerProduct{\Psi_\theta}{\Psi_\theta}},
\end{equation}
where $\op{H}$ is the (Hermitian) Hamiltonian of the quantum many-body system at hand and $\theta$ are the variational parameters of the ansatz. Equation \eqref{eq:cost_functional_nqs} is interchangeably denoted the name \emph{Energy}. We begin by rewriting the NQS ansatz to isolate the linear and non-linear features. For a spin configuration $S = \lbrace \sigma_j^z = \pm1\rbrace_{j = 1}^N$, standard expression in the literature for the trial wave-function $\Psi_\theta$ being a NQS with a Restricted Boltzmann machine (RBM) architecture is \cite{Carleo:2017nvk}
\begin{equation}
    \psi_{\mathrm{RBM}} (S) = \sum_{\lbrace h_i \rbrace}\expfunc{\sum_{j = 1}^N a_j\sigma_j^z + \sum_{i = 1}^L b_i h_i + \sum_{i = 1}^L \sum_{j=1}^N W_{ij}h_i \sigma_j^z},
\end{equation}
where here $h_i \in \lbrace \pm1\rbrace$ is a set of $L$ hidden spin variables. Since the RBM does not offer any intra-layer interactions, then \cite{Carleo:2017nvk}
\begin{align}
    \sum_{\lbrace h_i \rbrace}\expfunc{\sum_{i = 1}^L b_i h_i + \sum_{i = 1}^L \sum_{j=1}^N W_{ij}h_i \sigma_j^z} = \prod_{i = 1}^L2\cosh \left(b_i + \sum_{j = 1}^N W_{ij} \sigma_j^z \right),
\end{align}
allowing us to write $\psi_{\mathrm{RBM}}$ as
\begin{equation}
    \psi_{\mathrm{RBM}}(S) = \expfunc{\sum_j a_j \sigma_j^z} \prod_{i = 1}^L 2\cosh \left(b_i + \sum_{j = 1}^N W_{ij} \sigma_j^z \right).
\end{equation}
or alternatively, in logarithmic form as \cite{Carleo:2017nvk}
\begin{equation}
    \ln \psi_{\mathrm{RBM}} = \sum_{j = 1}^N a_j \sigma_j^z + \sum_{i = 1}^L \ln \left( 2\cosh \left(b_i + \sum_{j = 1}^N W_{ij} \sigma_j^z \right)\right).
\end{equation}
For simplicity, define a real feature vector
\begin{align}
    f_\vartheta(S) = (\sigma_1^z, \cdots, & \sigma_N^z,  \ln 2\cosh (b_1 + \sum_{j=1}^N W_{1j} \sigma_j^z),  \cdots , \ln 2\cosh (b_L + \sum_{j=1}^N W_{Lj} \sigma_j^z ))^\top \in \mathbb{R}^{h \times 1},
\end{align}
where $h = N + L$\footnote{Note that we will use the notation $\mathbb{R}^{1\times N}$ and $\mathbb{R}^{N \times 1}$ to make it explicit that we are using ``row" or ``column" vectors, respectively, in $\mathbb{R}^N$.}. Now, define the linear head
\begin{equation}
    \alpha = (a_1, \cdots, a_N, 1, 1, \cdots, 1) \in \mathbb{R}^{1 \times h}.
\end{equation}
This allows us to then write $\psi_{\mathrm{RBM}}$ as
\begin{equation}
    \ln\psi_{\mathrm{RBM}}(S) = \alpha \cdot f_\vartheta(S) \implies \psi_\mathrm{RBM} = \expfunc{\alpha \cdot f_\vartheta(S)}
\end{equation}
Note that in the equation above, in general one may have a $+ \beta$ term in the exponential. In what follows, we will denote by $f_\vartheta$ the ``trunk" (non-linear features) and by $(\alpha, \beta)$ the linear head parameters. Now for any architecture, if the trunk is real valued, one can generally write
\begin{equation}
\label{eq:STSHansatz}
    \psi_\theta(x) = \expfunc{\alpha \cdot f_{\theta^t}(x) + \beta} \expfunc{\ii\varphi \cdot f_{\theta^t}(x) + \ii\gamma},
\end{equation}
where $\theta := (\theta^t, \alpha, \varphi, \beta, \gamma)$ are the free parameters. Here, $\alpha, \varphi \in \mathbb{R}^{1 \times h}$ and $\beta, \gamma \in \mathbb{R}$. The trunk $f_{\theta^t}$ is simply a map from the configuration space to some $\mathbb{R}^{h\times 1}$ $f_{\theta^t} : \configSpace \rightarrow \mathbb{R}^{h\times 1}$ with trunk parameters denoted by $\theta^t$ ($t$ for trunk). Generally, one may allow for a complex valued trunk or even allow for two independent trunks for the phase and amplitude parts of the ansatz, but we will continue with such a common real valued trunk for simplicity. In this case, the number of heads is one (number of linear readout layers, or equivalently the number of amplitudes produced by a given $f_{\theta^t}$). Given that there is also one trunk only, we denote this approach as a \emph{single-trunk single-head} (ST-SH) class.
\newpara
Standard VMC methods would then update the free parameters of the network such that the cost is minimised. This is done efficiently by using local estimators such that \cite{Lange:2024nsr,Hibat_Allah_2020}
\begin{equation}
    \label{eq:cost_fn_nqs}
    E = \sum_x p_\theta(x) E_{loc}(x),
\end{equation}
where 
\begin{equation}
   p_\theta(x) := \frac{|\psi_\theta(x)|^2}{\sum_y |\psi_\theta(y)|^2} \farcomma E_{loc}(x) := \sum_y \op{H}_{xy}\frac{\psi_\theta(y)}{\psi_\theta(x)},
\end{equation}
are the Born probability and the local estimator of $\op{H}$ respectively. The gradients of the energy can be shown to be \cite{Lange:2024nsr,Hibat_Allah_2020}
\begin{equation}
\label{eq:gradSTSH}
    \partial_{\theta_i} E = 2\Re [\expect{(\mathcal{O}_\theta^i - \expect{\mathcal{O}_\theta^i})^*(E_{loc} - \expect{\op{H}}_\theta)}],
\end{equation}
where $\mathcal{O}_\theta^i(x) := \partial_{\theta_i} \ln \psi_\theta(x)$, with $i = 1, \cdots, P$, is the log derivative, where $P$ is the total number of free variational parameters. Once computed, the network parameters are updated using some descent method until a minimum for $C$ is reached.


\subsection{\label{sec:expressivity}Approximating multiple-states}

One naive way to use this standard VMC prescription outlined above to obtain all $D$ ground states in a degenerate system would require (at best) running $D$ simulations with different seeds. Even then, there is no guarantee that two obtained states from two independent simulations are different eigenstates. Thus, this sequential approach of training several ST-SH trial NQS wave-functions independently is not the most suitable for degenerate systems.
\newpara
Accordingly, among different algorithms it was proposed to consider an ensemble of $K$ many arbitrary trial wave-functions and to modify the cost function to include a penalty term that enforces orthogonality \cite{Wheeler_2024}. In the language of NQS, this would mean that one has $K$-many heads (amplitudes) each arising from equally $K$-many trial wave-functions, each with their own trunk (e.g. $K$-many independent networks for each target state). For uniformity, we denote this proposed approach in the current context as a \emph{multi-trunk multi-head} (MT-MH) ensemble. In that context, the single states in equation \eqref{eq:STSHansatz} are labelled by $k = 1, \cdots, K$ and thus the $k$\textsuperscript{th} amplitude is obtained from the $k$\textsuperscript{th} NQS via
\begin{equation}
\label{eq:MTMHansatz}
    \psi_{\pi_k} (x) = \expfunc{\alpha_k \cdot f_{\theta_k^t}(x) + \beta_k} \expfunc{\ii\varphi_k \cdot f_{\theta_k^t}(x) + \ii\gamma_k}.
\end{equation}
where $\pi_k$ is the set of all variational parameters $(\alpha_k, \beta_k, \varphi_k, \gamma_k, \theta_k^t)$ for the target state $k$ which form a disjoint vector $\pi_k \in \mathbb{R}^{P_k}$ and hence and the total parameter set for the ensemble is $\Theta = (\pi_1, \cdots, \pi_K) \in \mathbb{R}^{P_{tot}}$. Note that here, $\theta_k^t$ does not denote the $k$\textsuperscript{th} trunk parameter, but rather the parameters of $k$\textsuperscript{th} trunk. The cost function \eqref{eq:cost_functional_nqs} then is modified to include an additional overlap and orthogonality penalty term $\mathcal{P}$ such that \cite{Wheeler_2024}
\begin{equation}
    \mathcal{P}(\Theta) := \frac{1}{2} \sum_{k \neq l} \frac{|\Sigma_{kl}|^2}{N_k N_l},
\end{equation}
whereby $\mathcal{P} = 0$ iff the normalised states are orthonormal. Here, $\Sigma_{kl}(\Theta) := \innerProduct{\psi_{\pi_k}}{\psi_{\pi_l}}$ define the overlaps and $N_k := \Sigma_{kk}$.
\newpara
For some weights $w_k$ such that $\sum_k w_k = 1$, the total ensemble cost is then
\begin{equation}
\label{eq:costMTMH}
    C(\Theta) := \sum_{k = 1}^K w_k \frac{\innerProduct{\psi_{\pi_k}}{\op{H} \psi_{\pi_k}}}{\innerProduct{\psi_{\pi_k}}{\psi_{\pi_k}}} + \lambda \mathcal{P}(\Theta),
\end{equation}
where $\lambda > 0$ is the penalty strength. This modification of the cost ensures that, for carefully chosen weights and penalty factor, one obtains $K$ many solutions which minimise the energy but are also mutually orthogonal. One can then obtain expressions for the gradients of this cost function using standard methods.
\newpara
The MT-MH NQS ensemble training, while successful in obtaining $K$ eigenstates from a degenerate lowest energy eigenspace, has a rather large computational footprint. We now propose an ensemble approach whereby we consider \emph{one} NQS with $K$-many linear read-outs, each targeting one of the $K$ degenerate eigenstates. As the read-outs share the same trunk, we therefore denote this as a \emph{single-trunk multi-head} (ST-MH) ensemble.  
\newpara
Consider now a \emph{shared} trunk $f_\vartheta : \configSpace \rightarrow \mathbb{R}^{h \times 1}$ with parameters $\vartheta \in \mathbb{R}^T$ where $T$ is the number of trunk parameters. For each head $k$, one has the \emph{head} parameters $\theta^{(h)}_k := (\alpha_k, \varphi_k, \beta_k, \gamma_k) \in \mathbb{R}^{P_H}$ with, as usual, $\alpha_k, \varphi_k \in \mathbb{R}^{1 \times h} \,,\, \beta_k , \gamma_k \in \mathbb{R}$. The set of all ensemble parameters is then $\Theta := (\vartheta, \theta^{(h)}_1, \cdots, \theta^{(h)}_K) \in \mathbb{R}^{P_T + KP_H}$. The $k$\textsuperscript{th} target eigenstate can be obtained from the head $k$ simply by
\begin{equation}
    \psi_{\pi_k}(x) = \expfunc{\alpha_k \cdot f_\vartheta(x) + \beta_k} \expfunc{\ii\varphi_k \cdot f_\vartheta(x) + \ii\gamma_k},
\end{equation}
where now $\pi_k = (\theta_k^{(h)}, \vartheta)$, which can be written in the compact form
\begin{equation}
    \psi_k(x) = \expfunc{\chi_k \cdot f_\vartheta(x) + c_k},
\end{equation}
where $\chi_k := \alpha_k + \ii\varphi_k \in \mathbb{C}^{1 \times h}$ and $c_k := \beta_k + \ii\gamma_k \in \mathbb{C}$. It can be shown (see Appendix \ref{app:VMCNQSClasses}) that one can obtain consistent expressions for the gradients of a cost function similar to \eqref{eq:costMTMH} enforcing the same penalty term $\mathcal{P}$ (see Appendix \ref{app:VMCNQSClasses}). 
\newpara
For the case of degenerate eigenspaces, one can (see Appendix \ref{app:sufficiencytheorem}) associate to any set of orthonormal eigenstates a linear modulus span $\mathcal{R}_G$ (span of all target log-moduli on a common support $\suppPsi$ where the target states have a positive modulus) whose dimension $r_G = \dim\mathcal{R}_G$ we call the \emph{linear modulus rank} of the degenerate manifold. Similarly for the phases, one can define the linear phase span $\mathcal{R}_\Omega$ (also on their common support) whose dimension $r_\Omega = \dim \mathcal{R}_\Omega$ we call the \emph{linear phase rank} of the degenerate manifold. This enables us to define a combined rank $r_\mathrm{both} = \dim \mathrm{span}\bigl(\mathcal{R}_G \cup \mathcal{R}_\Omega\bigr)$ which is the combined \emph{linear rank} of the states' log-moduli and phases on the common support in the degenerate target manifold.
\newpara
As shown in Theorem \ref{thm:sufficiency} (Appendix  \ref{app:sufficiencytheorem}), the representational capacity of a ST-MH NQS ensemble is governed precisely by this rank. If $r_\mathrm{both} \leq h + 1$, where $h$ is the latent width of the shared trunk and $r_\mathrm{both}$ is computed from the chosen single-valued phase branches on $\suppPsi$, then the ST-MH ensemble has the capacity to exactly represent \emph{all} $D$ degenerate eigenstates on $\suppPsi$. Conversely, if $r_\mathrm{both} > h + 1$, then no single trunk of width $h$ suffices to represent the entire degenerate manifold. As such, using the orthogonality penalty outlined above, one can use a single NQS wave-function which has $D$-many lightweight linear-heads to obtain all $D$-many degenerate ground states, instead of using $D$-many independent NQS wave-functions, when the criterion above is met.


\subsection{\label{sec:numericalstudy}Numerical proof of principle}

To demonstrate the applicability of the ST-MH NQS ensemble, we compare its performance to that of the MT-MH ensemble and further demonstrate true full ground space resolution. We first present the model which we will consider in this work, described in Section \ref{sec:physicalmodel}. Next, we provide a numerical proof of principle which is performed over two steps in Sections \ref{sec:firstpop} and \ref{sec:extension}, the first of which is a computational efficiency comparison between the two ensemble approaches. The second is the verification of the true full ground space resolution of the considered model. 


\subsubsection{\label{sec:physicalmodel}Physical model}
We select a model which conforms to the bounds of the Theorem \ref{thm:sufficiency} and its requirements, namely the frustrated spin-$\tfrac{1}{2}$ $J_1-J_2$ Heisenberg chain on a periodic ring with even $N$ sites defined by the Hamiltonian
\begin{equation}
\label{eq:heisenberghamiltonian}
    \op{H} = J_1 \sum_{j = 1}^N \vec{S}_j \cdot \vec{S}_{j+1} + J_2 \sum_{j = 1}^N \vec{S}_j \cdot \vec{S}_{j+2},
\end{equation}
where $\vec{S}_j = (S_j^x, S_j^y, S_j^z)$ and $\vec{S}^2_j = 3/4$ and the site indices are understood modulo $N$. Throughout, we work in a fixed $S^z$ sector. For even $N$, the ground space then lies in $S^z_\mathrm{tot} = 0$ whose computational basis is of $\binom{N}{N/2}$ dimensions.
\newpara
At the Majumdar-Ghosh (MG) point $J_2 = J_1 / 2$ \cite{C_K_Majumdar_1970,10.1063/1.1664978,10.1063/1.1664979}, the model admits an exact solution. Introducing 3-site total spin operators $\vec{\tau}_j = \vec{S}_j + \vec{S}_{j+1} + \vec{S}_{j+2}$, one can see that summing the direct expansion of $\vec{\tau}^2$ over $j$ on the ring counts each on-site Casimir three times: each nearest-neighbour pair four times (each nearest-neighbour pair belongs to two adjacent triples and, in each, enters with the cross-term prefactor 2, yielding an overall factor of 4) and each next-nearest neighbour pair twice. Inserting that into \eqref{eq:heisenberghamiltonian}, at the MG point this yields the positive-semidefinite representation
\begin{equation}
    \label{eq:mgheisenberg}
    \op{H}_\mathrm{MG} = \frac{J_1}{4}\sum_{j=1}^N \vec{\tau}^2_j - \frac{9J_1}{16}N.
\end{equation}
Equivalently, one can decompose each triple into $S_\mathrm{tot} = \tfrac{1}{2}, \tfrac{3}{2}$ sectors. Since $\vec{\tau}_j^2 = 3/4 + 3\Pi_j^{(3/2)}$ where $\Pi^{(3/2)}_j$ being the projector onto total spin $\tfrac{3}{2}$ on $(j, j+1, j+2)$, one can write
\begin{equation}
\label{eq:projectedmg}
    \op{H}_\mathrm{MG} = \frac{3J_1}{4}\sum_{j = 1}^N \Pi_j^{(3/2)} - \frac{3J_1}{8}N \farcomma \Pi_j^{(3/2)} = \frac{1}{3}\left(\vec{\tau}_j^2 - \frac{3}{4} \right).
\end{equation}
Hence, the ground energy is bounded from below by $-3J_1N/8$, with equality iff the state lies in the kernel of every $\Pi_j^{(3/2)}$. 
\newpara
The two exact ground states are obtained as period-2 products of nearest neighbour singlets. Writing
\begin{equation}
    \ket{s_{i, i+1}} = \frac{1}{\sqrt{2}} \left(\ket{\uparrow_i \,\,\, \downarrow_{i+1}} - \ket{\downarrow_i \,\,\, \uparrow_{i+1}}\right),
\end{equation}
with $S_i^z \ket{\uparrow_i} = 1/2 \ket{\uparrow_i}, S_i^z \ket{\downarrow_i} = -1/2 \ket{\downarrow_i}$, one can then define
\begin{equation}
    \ket{\Phi_A} = \bigotimes_{m = 1}^{N/2} \ket{s_{2m-1, 2m}} \farcomma \ket{\Phi_B} = \bigotimes_{m = 1}^{N/2} \ket{s_{2m, 2m+1}}.
\end{equation}
Each 3-site block $(j, j+1, j+2)$ inside either product contains exactly one singlet bond together with a decoupled spin-$\tfrac{1}{2}$, hence its total spin is $S_\mathrm{tot} = \tfrac{1}{2}$ and $\Pi_j^{(3/2)} \ket{\Phi_{A/B}} = 0$ for all $j$. Consequently, 
\begin{equation}
    \op{H}_\mathrm{MG} \ket{\Phi_{A/B}} = -\frac{3J_1}{8}N \ket{\Phi_{A/B}} \farcomma E_0(N) = -\frac{3J_1}{8}N,
\end{equation}
and thus the two states saturate the lower bound in equation \eqref{eq:projectedmg}. Translations by one site, $\op{T} \vec{S}_j \op{T}^\dagger = \vec{S}_{j+1}$, exchanges two coverings ($\op{T}\ket{\Phi_A} = \ket{\Phi_B}$ and $\op{T}\ket{\Phi_B} = \ket{\Phi_A}$). On the periodic ring the ground manifold is thus the 2-dimensional subspace $\mathcal{G} = \mathrm{span}\lbrace \ket{\Phi_A}, \ket{\Phi_B} \rbrace$, and it is natural to resolve it into crystal-momenta eigenstates. The one-site translation restricted to $\mathcal{G}$ has eigenvalues $\pm 1$. One can choose normalised representatives as
\begin{equation}
\label{eq:most}
    \ket{\Psi_{k_\pm}} = \frac{\ket{\Phi_A} \pm \ket{\Phi_B}}{\sqrt{2(1 \pm \expect{\Phi_B | \Phi_A})}} \farcomma \op{T} \ket{\Psi_{k_\pm}} = e^{ik_\pm}\ket{\Psi_{k_\pm}},
\end{equation}
with $k_+ = 0, k_- = \pi$. For even $N$, the overlap of the two coverings is elementary, $\expect{\Phi_B | \Phi_A} = (-1)^{N/2} 2^{1-N/2}$, thus the normalisation above is explicit. Both $\ket{\Phi_{A/B}}$, and hence $\ket{\Psi_{k_\pm}}$, are total-spin singlets and belong to the $S^z_\mathrm{tot} = 0$ sector.
\newpara
From the support point of view, one can in the computational basis label states such that $\ket{\mathbf{s}} = \ket{s_1 \cdots s_N}$ denotes the $S^z$-product states with $s_j \in \lbrace \downarrow, \uparrow \rbrace$ and $S_\mathrm{tot}^z = 0$. One can see that $\mathrm{supp}\Phi_A$ and $\mathrm{supp}\Phi_B$ each has a cardinality of $2^{N/2}$. One either support, the amplitudes have flat modulus $2^{-N/4}$ (a product of $1/\sqrt{2}$ per singlet) and bond-dependent phases determined by the singlet orientations. Off of these supports, the amplitudes vanish identically. The two supports intersect only on the two Néel configurations, and thus $|\mathrm{supp}\Phi_A \cup \mathrm{supp}\Phi_B| = 2^{1 + N/2} - 2$. 
\newpara
The momentum eigenstates live on the union support and differ by relative phases between the two coverings. In particular, on the two Néel configurations, the relative phase between $\ket{\Phi_A}$ and $\ket{\Phi_B}$ is $(-1)^{N/2}$, and one of the combinations in \eqref{eq:most} has an exact node there, namely
\begin{equation}
    \expect{\text{Néel}|\Psi_{k_+}} = 0 \,\,\,\text{if } N = 2 (\mathrm{mod}4) \farcomma \expect{\text{Néel}|\Psi_{k_-}} = 0 \,\,\,\text{if } N = 0 (\mathrm{mod}4),
\end{equation}
while away from these two configurations $\ket{\Psi_{k_-}}$ and $\ket{\Psi_{k_+}}$ have nonzero amplitudes. Altogether, the common support $\mathrm{supp}\Psi_{k_+} \cap \mathrm{supp}\Psi_{k_-}$ equals $\mathrm{supp}\Phi_A \cup \mathrm{supp}\Phi_B$ with the two Néel strings removed.
\newpara
The representation \eqref{eq:projectedmg} ensures that the states described above are exact eigenstates of $\op{H}_\mathrm{MG}$, not merely variational minima, and fixes the ground energy density to $E_0 / N = -3J_1/8$. For the purpose of symmetry resolution on the periodic ring, it is natural to work with the translation eigenstates $\ket{\Psi_{k_\pm}}$ carrying crystal momenta $k = 0, \pi$. The dimer products $\ket{\Phi_{A/B}}$ are simply symmetry-broken representatives of the same 2-dimensional ground manifold and are mapped into one another by a 1-site translation.


\subsubsection{\label{sec:firstpop}ST-MH versus MT-MH empirical efficiency and accuracy}
The first part will concern the performance analysis (both in terms of computational resources and optimisation accuracy) of the ST-MH approach compared to that of the MT-MH approach. Throughout, we set $J_1 = 1$ and hence $J_2 = 0.5$. Further, the trunk for both the ST-MH and MT-MH ensembles is composed of a fully connected 2-hidden layer multi-layer perceptron (MLP) whereby both hidden layers have the same fixed width and one ReLU activation layer sits between them. The orthogonality penalty strength is started from an initial value of $\lambda_s = 1\times 10^{-3}$ and annealed to $\lambda_f = 1.0$ over 100 optimisation steps. Further, we use an Adam optimiser \cite{adam} with a fixed learning rate of $1\times 10^{-3}$. Samples are generated using a Metropolis-Hastings sampler \cite{10.1063/1.1699114} with a transition kernel which proposes new configurations by exchanging the spins of two sites. The initial proposal is created within the $S^z$ sector and hence, this transition kernel preserves the sector. The number of sweeps for the sampler is chosen to be 5 with a total of 512 samples generated at each sampling step distributed over 8 Markov chains. Unless mentioned otherwise, the ensembles use a shared mixture sampling prescription (see Appendix \ref{app:methodology}). For the following, we set $N = 4$, and consequently, $E_0(N=4) = -1.5$. As the first part concerns performance metrics, \emph{full} eigenspace resolution for different values of $N$ is discussed in Section \ref{sec:extension} below.
\newpara
We note that as this current work merely serves as a proof-of-principle, the purpose of this work is therefore \emph{not} to find an optimal network architecture to efficiently solve this physical model but rather focus on the applicability of the ST-MH ensemble approach. Therefore, at no point do we concern ourselves with, for example, choosing an architecture which has fewer number of parameters than there are states in the space. However, we demonstrate through brief ablation studies in Section \ref{sec:extension} that one can, as one expects, approach the same problem with much smaller network sizes.
\newpara
We now demonstrate that using the ST-MH approach, one obtains solutions as accurately, as quickly, but more efficiently compared to the MT-MH approach. The first comparison done is observing both the runtime and the number of parameters as the number of target states $K$ increases. Figure \ref{fig:params_vs_k} shows the parameters count for different $K$. Keeping $N = 4$ and $h = 32$ for both ensembles, we grow the number of heads for the ST-MH and duplicate accordingly the wave-functions for MT-MH ensembles. The measured counts match a linear-in-$K$ scaling for both ST-MH and MT-MH parameter counts (MT-MH duplicates the trunk per head). 

\begin{figure}[h]
    \centering
    \includegraphics[width=0.6\linewidth]{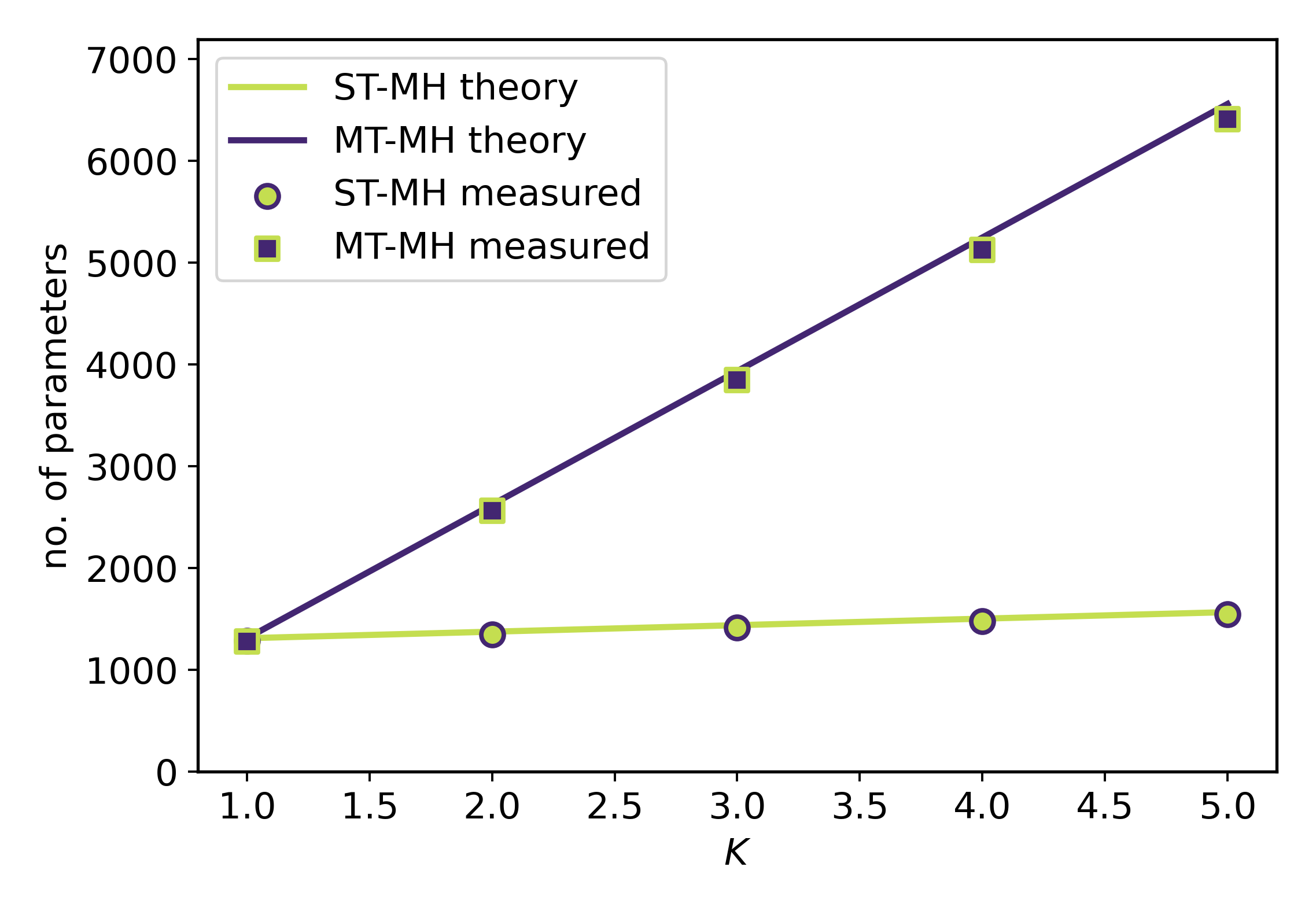}
    \caption{The number of parameters observed compared to the theoretical prediction for both the ST-MH and MT-MH NQS ensemble training for a network architecture composed of a MLP with 2 hidden layers and ReLU activations.}
    \label{fig:params_vs_k}
\end{figure}
\noindent
Here, the theory lines shown in the figure correspond to the estimates conducted in equations \eqref{eq:paramscountth} and \eqref{eq:paramscountth2} in Section \ref{sec:costandparamscount}. The substantially fewer number of parameters for the ST-MH NQS ensemble implies not only a smaller memory footprint but also computational time. This can be seen in Figure \ref{fig:runtime_vs_k}.

\begin{figure}[h]
    \centering
    \includegraphics[width=0.6\linewidth]{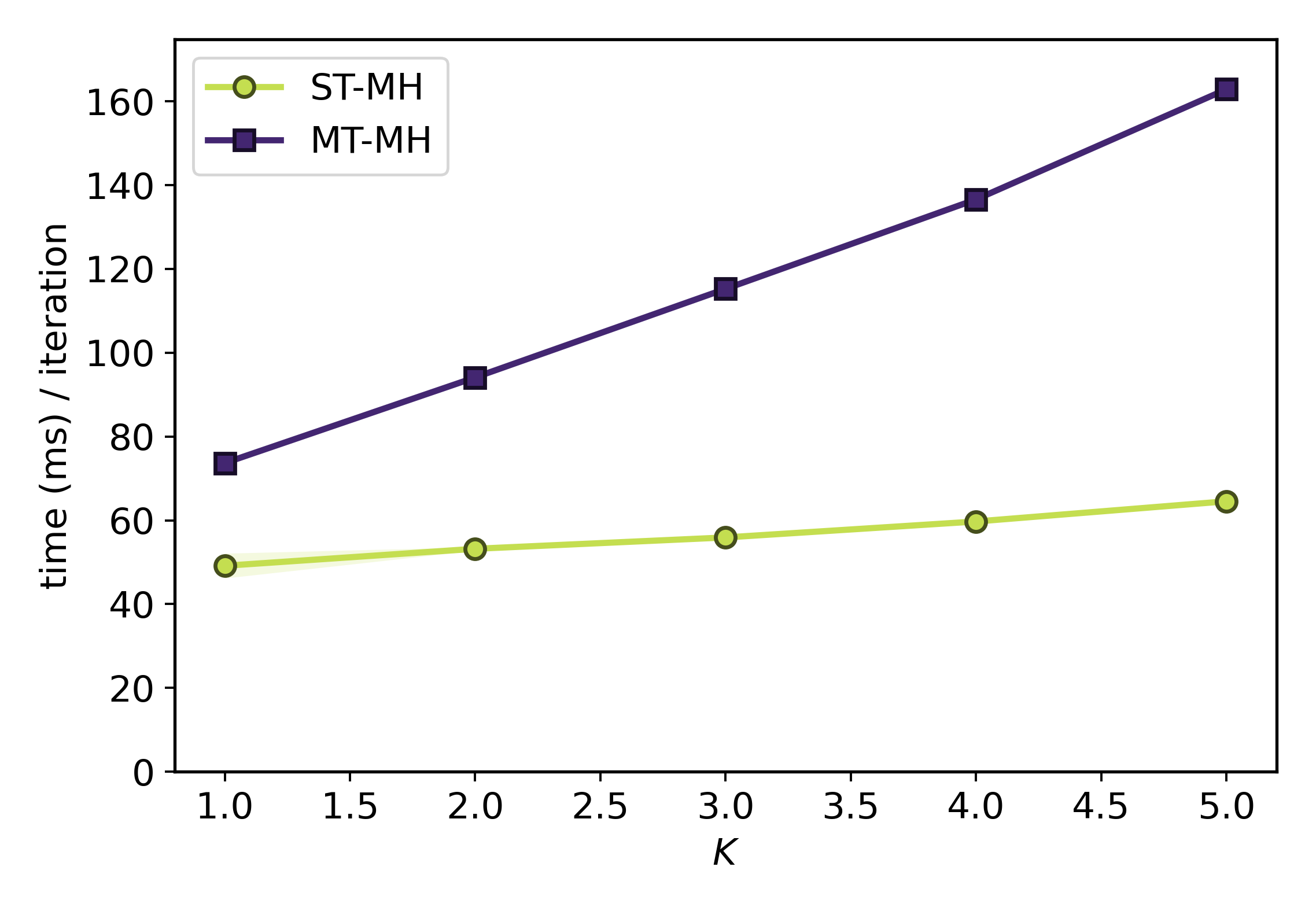}
    \caption{The average seconds-per-iteration time with respect to $K$ for both ST-MH and MT-MH ensembles with the sampler cost held fixed averaged over 3 simulations for each $K$, whereby each included 100 optimisation.}
    \label{fig:runtime_vs_k}
\end{figure}
\noindent
Figure \ref{fig:runtime_vs_k} shows the scaling of the average seconds-per-iteration with respect to different values of $K$. Here, the mean is computed over 3 simulations, each with 100 optimisation steps. The figure shows the case for a fixed sampler cost. Namely, both ensembles draw from a shared mixture (instead of having independent samplers per-head/state). The reason for this is to make the comparison explicit. For small networks/systems, the compute time required for sampling amortises that of the gradient computations, resulting in skewed results. 
\newpara
In this case, one sees that the ST-MH ensemble compute time is nearly flat with $K$ (only head and overlap compute resources grow) whereas MT-MH ensemble compute time increases roughly linearly due to the $K$-fold replication of the trunks. The residual slope in both curves is explained by the shared $O(K^2)$ pairwise-overlap computations in the penalty. More importantly, the smaller overall computational footprint for the ST-MH NQS ensemble does not come at a cost of accuracy, as can be seen in what follows.

\begin{figure}[h]
    \centering
    \includegraphics[width=0.6\linewidth]{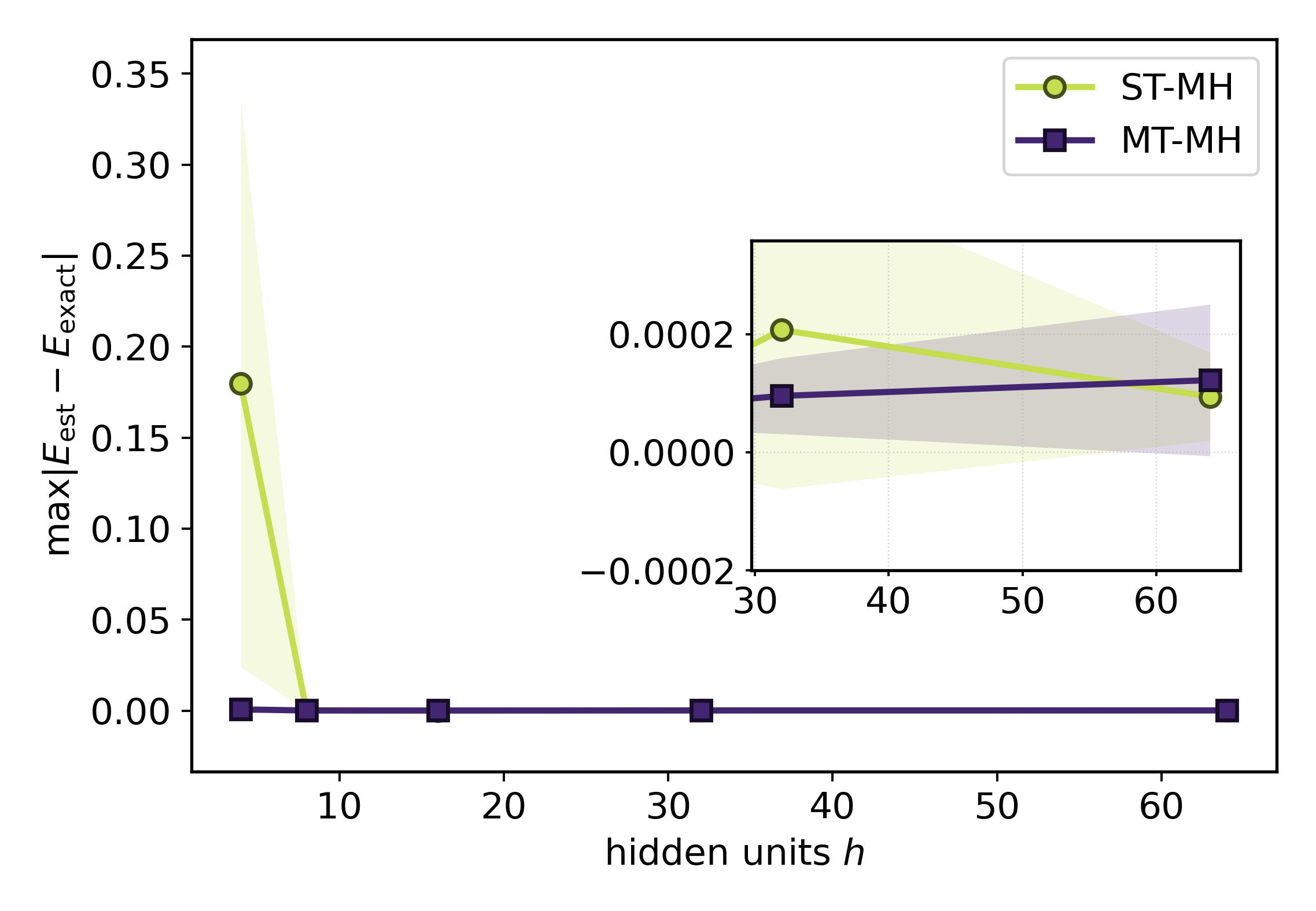}
    \caption{Maximum absolute energy error across heads as a function of hidden units is shown. Both ST-MH and MT-MH ensembles achieve high and similar accuracy especially for higher trunk width. The results are averaged over 3 simulations with 100 optimisation steps each.}
    \label{fig:error_vs_width}
\end{figure}
\noindent
Figure \ref{fig:error_vs_width} shows the maximum absolute energy error across heads as a function of hidden units $h$. The results are averaged over 3 simulations with 100 optimisation steps each. As shown, for relatively low $h$, the MT-MH approach has higher accuracy, unsurprisingly, due to the fact that each target degenerate eigenstate has $h$ trunk features in total to represent it. Comparatively, the $h$ trunk features are shared among all target degenerate eigenstates in the ST-MH case. This, however, is not an unavoidable hindrance, as is shown in the ablation studies in Section \ref{sec:extension}.
\newpara
Irrespectively, as $h$ increases, it is evident that the maximum absolute energy error of ST-MH quickly declines and further matches that of the MT-MH, indicating sufficient representability of the total target states with a fraction of the available feature space. The shown plateau in both cases is an indicator that the network has become expressive enough to saturate the accuracy allowed by the optimisation budget. One can also plot the obtained energy eigenvalue for both the heads in each ensemble as shown below.

\begin{figure}[h]
    \centering
    \includegraphics[width=0.6\linewidth]{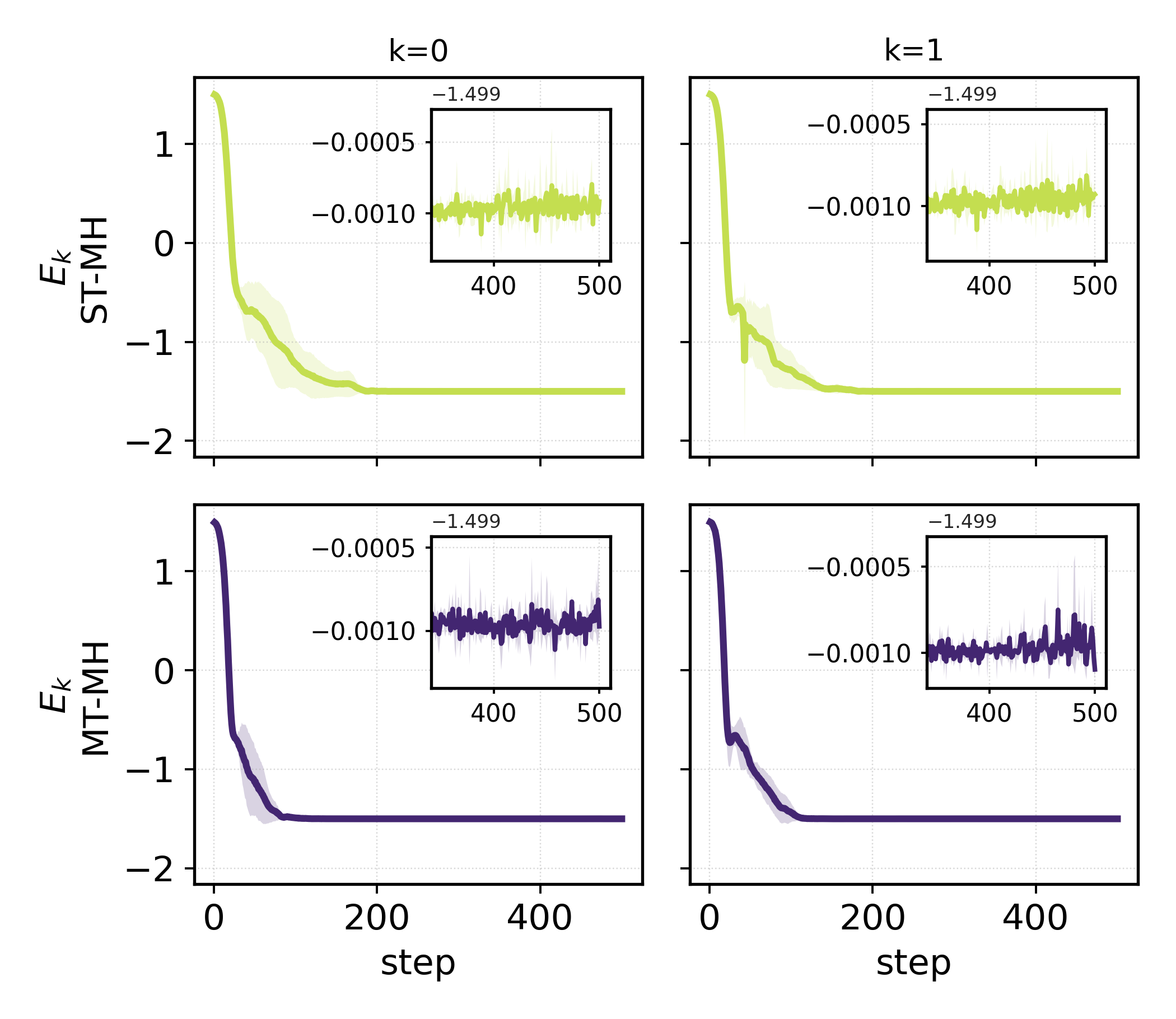}
    \caption{The evolution of the two eigenvalues $E_k$ obtained by the two heads for the case of ST-MH and MT-MH ensembles. The true ground energy is $E_0 = -1.5$. The results are averaged over 3 simulations with 500 optimisation steps each.}
    \label{fig:learning_E}
\end{figure}
\noindent
Figure \ref{fig:learning_E} shows a plot of the evolution of the two eigenvalues that both the two heads/networks in the ST-MH and MT-MH ensembles, respectively, converge to. Once more, the shown results are averages over 3 simulations, now each with 500 optimisation steps. As shown in the figure, both ensemble approaches converge quickly to the correct ground energy $E_0 = -1.5$ and stabilise near the same plateau, showing identical \emph{learning speed} despite their very different parameter counts. The fluctuations observed for both ensembles are minor and are attributed to simply optimisation noise. They are neither large enough to cause any discrepancy, nor large enough to cause any instability in training. Lastly, to demonstrate that both ensembles obtained different states, the evolution of the overlap matrix norm is shown in Figure \ref{fig:learning_S}.

\begin{figure}[h]
    \centering
    \includegraphics[width=0.6\linewidth]{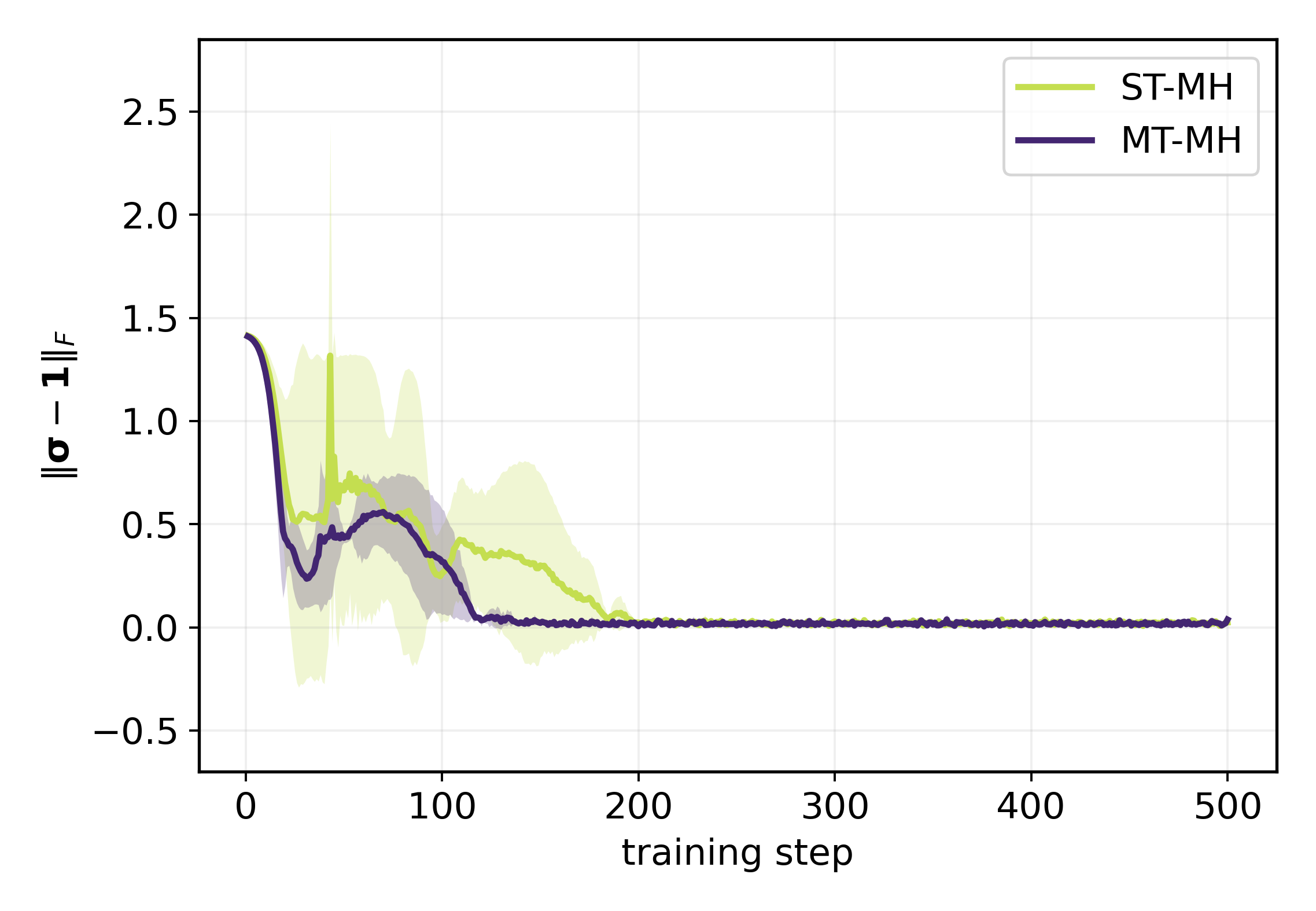}
    \caption{The Frobenius deviation $\lVert \sigma - \mathbf{1}\rVert_F$ is shown for both the ST-MH and MT-MH ensembles. Both ensembles rapidly reduce the deviation and settle near zero. The results are averaged over 3 simulations with 500 optimisation steps each.}
    \label{fig:learning_S}
\end{figure}
\noindent
Figure \ref{fig:learning_S} shows the evolution of the orthogonality metric, namely the Frobenius norm $\lVert \sigma - \mathbf{1}\rVert_F$, where $\mathbf{\sigma}$ is the normalised overlap matrix between heads and $\mathbf{1}$ is the identity. Note that here, the pairwise normalised overlap $\sigma_{kl}$, of which $\sigma$ is constructed from, is computed using the biased estimator outlined in Appendix \ref{app:methodology} equation \eqref{eqapp:estimatorsigma}. However, post-training (see Table \ref{tab:groundspace-metrics} below), an exact computation is carried out whereby the full-space is enumerated and pairwise-overlaps are computed explicitly. It was observed to be the case that the exact computation fell within low error ranges. 
\newpara
As shown in the figure, both ensembles converge to nearly 0, indicating that the two obtained normalised states are mutually orthogonal. Coupled with the fact that the two obtained states, in each ensemble approach, independently also arrived at an energy eigenvalue close to the true $E_0$, this is a \emph{soft} indicator which may allow one to conclude that the two are indeed different eigenstates in the ground space. In the following section, we ensure that this is indeed the case. Consequently, this implies that the orthogonality penalty did in fact force the states to be mutually orthogonal as intended.
\newpara
Lastly, we note that the compute time as a function of the system size $N$ was also studied for different chain lengths of even sites $N = 4, \cdots, 10$ and $K=2$. It was observed that ST-MH remains faster than MT-MH at all sizes, but both compute times grow modestly. This was true for both ensembles for both the shared sampler approach or the independent sampler approach where each head received its own sampler. 
\newpara
We note that in principle, for the MT-MH ensemble, there could be some trunk width $h$ which is smaller than that considered for the ST-MH ensemble but still converges to the solutions. In that case, the total parameter count and the total runtime will differ from what is presented here. However, as this specific trunk width may be difficult to know a priori, for the purpose of comparison we set both trunk widths to be the same for the ST-MH and MT-MH ensembles. A more detailed discussion on different trunk widths is presented in Section \ref{sec:cost} below.


\subsubsection{\label{sec:extension}Fidelity tests and full ground space resolution}

The purpose of this section is to demonstrate the capability of the ST-MH NQS ensemble to resolve the entire degenerate ground space. To test this, the NQS ensemble is used to resolve the degenerate ground space of once again the same frustrated $J_1-J_2$ Heisenberg chain at the MG point described in previous sections, but with various number of sites $N = 4, 6, 8$. We also present brief ablation studies for the case of $N=4$. To this end, we evaluate a series of diagnostics that jointly assess both the accuracy of individual heads as approximate eigenstates and the collective ability of the ensemble to span the full degenerate manifold. 
\newpara
First, post-training and using full enumeration (e.g. obtain the $k$\textsuperscript{th} obtained solution by using the optimised ST-MH ensemble to output the amplitudes for all basis states in the Hilbert space through head $k$), for each head we compute the expectation value of the Hamiltonian $\op{H}$ and its variance $\mathrm{Var}(\op{H})$. The expectation provides a direct comparison to the exact ground energy, while the variance quantifies how close the state is to being a true eigenstate. Next, we measure the fidelity of each head with the exact ground state subspace, as obtained from exact diagonalisation. The ground-subspace fidelity captures the fraction of the head's wave-function that resides within the true ground manifold. Averaging this quantity across all heads provides a compact insight into ensemble accuracy, while the spread across heads indicates how uniformly the states approach the ground space.
\newpara 
Lastly, in order to truly attest that the ensemble does not merely produce orthogonal states in the full Hilbert space but actually spans the ground subspace itself, we project the heads onto the exact ground space and analyse the resulting projection matrix. To formalise this, let $\Psi = [\psi_1 \,\,\,\dots\,\,\,\psi_K]$ denote the matrix of obtained solutions generated from the $K$ optimised heads (each column representing a normalised variational wave-function in the computational basis). Let $V_0 = [v_1\,\,\,\dots\,\,\,v_g]$ be the exact ground state eigenvectors obtained by exact diagonalisation, spanning the degenerate manifold of dimension $g$. The orthogonal projector onto the ground space is then $P_{\mathcal{G}} = V_0 V_0^\dagger$. With this notation, we define the projection matrix
\begin{equation}
    \mathscr{C} = V_0^\dagger \Psi,
\end{equation}
whose entries encode the overlaps $\langle v_i | \psi_k \rangle$ between the exact ground basis and the learned heads. Several diagnostics are derived from $C$:
\begin{itemize}
    \item \emph{Ground-subspace fidelity:} for a head $\psi_k$, the fidelity is given by
    \begin{equation}
        F_k = \langle \psi_k | P_\mathcal{G} |\psi_k\rangle,
    \end{equation}
    namely the squared norm of its projection onto the ground manifold. Values lie in $[0,1]$, with $F_k=1$ signifying perfect confinement to the ground subspace and $F_k=0$ indicating full leakage into excited states.

    \item \emph{Pairwise overlaps:} the head overlap matrix is defined as
    \begin{equation}
        \sigma_{kl} = \langle \psi_k | \psi_l \rangle,
    \end{equation}
    and computed via exact computation \emph{not} the biased estimator as done during training and shown in Figure \ref{fig:learning_S}. Ideally, $\sigma$ approaches the identity, reflecting mutual orthogonality of the obtained states. Deviations, as done previously, are quantified by the Frobenius norm $\lVert \sigma - I \rVert_F$, where values near zero indicate nearly orthogonal solutions. 

    \item \emph{Singular values and principal angles:} performing singular-value decomposition (SVD) of the projection matrix, $\mathscr{C} = U \Sigma W^\dagger$, yields non-negative singular values $\sigma_i \in [0,1]$ whose cosines are the principal angles between the subspace spanned by the ensemble $\Psi$ and the exact ground space. Several conclusions can be drawn:
    \begin{enumerate}
        \item $\sigma_i = 1$ (angle $0^\circ$) means that direction is perfectly captured,

        \item $\sigma_i = 0$ (angle $90^\circ$) means that direction is entirely missing. Intermediate values quantify partial coverage.

        \item Rank condition: if $\mathrm{rank}(\mathscr{C}) = g$, then all directions of the ground space are in principle represented by the ensemble. Missing rank indicates that some eigenstates are absent.

        \item Condition number: the ratio $\kappa = \sigma_{\max}/\sigma_{\min}$ characterises numerical stability. A moderate $\kappa$ implies the ensemble spans the manifold in a well-conditioned manner. Large $\kappa$ values reflect near-linear dependence of the projected heads, meaning some directions are only weakly resolved even if rank is complete.
    \end{enumerate}

    \item \emph{Effective dimension $d_\mathrm{eff}$:} the effective dimensions are the dimensions of the ground space which we considered as truly resolved. This is accounted for by counting the number of singular values of the SVD decomposed projection matrix which are $\geq 0.99$.
\end{itemize}
Together, these metrics provide a layered diagnostic: energy and variance measure individual accuracy, fidelities confirm that each state lies within the ground space, the overlap matrix quantifies orthogonality and the SVD of $\mathscr{C}$ attests that the ensemble as a whole spans the degenerate ground manifold with controlled conditioning.

\begin{table*}[h]
\centering
\small
\setlength{\tabcolsep}{4pt}
\renewcommand{\arraystretch}{1.15}
\caption{\label{tab:groundspace-metrics}Representative diagnostics (across 3 simulations) of ST-MH ensembles for resolving the degenerate ground space of the frustrated $J_1-J_2$ Heisenberg chain at the MG point with $N=4, 6, 8$ sites. Listed are the ground state degeneracy $g$, number of ensemble heads $K$, exact and average learned head-energies $E_0$ and $\bar{E}$, maximum per-head quantum variance of the full Hamiltonian at the final iterate across all heads, mean and minimum ground-subspace fidelity $F_\mathrm{mean}, F_\mathrm{min}$, rank of the projection matrix $\mathscr{C}$ relative to $g$, smallest singular value $\sigma_{\min}$, condition number $\kappa$, Frobenius deviation of the head overlap matrix from the identity, and the lowest effective dimension $d_{\mathrm{eff}}$ resolved by the ensemble in 3 conducted simulations as an independent column. Ablation studies for the $N = 4$ case are marked with (B) and (C).}
\resizebox{\textwidth}{!}{
\begin{tabular}{cccccccccccc|c}
\hline
$N$ & $g$ & $K$ & $E_0$ & $\bar{E}$ & $\max\mathrm{Var}(\hat{H})$ & $F_{\mathrm{mean}} \pm \mathrm{std}$ & $F_{\min}$ & $\mathrm{rank}(C)/g$ & $\sigma_{\min}(C)$ & $\kappa(C)$ & $\|\sigma - \mathbf{1}\|_F$ & $d_{\mathrm{eff}}$ \\
 &  &  &  &  & ($\times 10^{-3}$) &  &  &  &  &  &  &  \\
\hline
4 & 2 & 2 & $-1.5000$ & $-1.5000$ & 0.035 & $0.999991 \pm 2.50\times 10^{-6}$ & $0.999988$ & 2/2 & $0.999991$ & 1.000 & 0.004736 & 2 \\ 
4(B) & 2 & 2 & $-1.5000$ & $-1.4999$ & 0.574 & $0.999928 \pm 5.35\times 10^{-5}$ & $0.999874$ & 2/2 & $0.999933$ & 1.000 & 0.008485 & 2 \\
4(C) & 2 & 2 & $-1.5000$ & $-1.5000$ & 0.076 & $0.999967 \pm 3.00\times 10^{-6}$ & $0.999964$ & 2/2 & $0.999976$ & 1.000 & 0.006870 & 2 \\
6 & 2 & 2 & $-2.2500$ & $-2.2492$ & 1.612 & $0.999319 \pm 1.01\times 10^{-4}$ & $0.999219$ & 2/2 & $0.999524$ & 1.000 & 0.004153 & 2 \\
8 & 2 & 2 & $-3.0000$ & $-2.9977$ & 6.143 & $0.998158 \pm 6.56\times 10^{-4}$ & $0.997502$ & 2/2 & $0.998601$ & 1.001 & 0.080933 & 2 \\
\hline
\end{tabular}
}
\end{table*}
\noindent
As seen in Table \ref{tab:groundspace-metrics}, across various system sizes of $N = 4, 6, 8$, the diagnostics indicate good convergence\footnote{Exact simulation setup parameters used are shown in Table \ref{tab:syssetup} in Appendix \ref{app:syssetup}} to the two degenerate momentum eigenstates. Across three seeds per $N$, the per-simulation mean energies are highly consistent (maximum standard deviation of $\approx 10^{-3}$ for $N = 6$). Likewise, the mean ground subspace fidelities were observed to also be stable (maximum standard deviation  $\approx 10^{-3}$ for $N = 6$) with a maximum standard deviation in the spread of $\approx 10^{-3}$ for $N = 6$). We therefore show a single representative per run $N$ for conservatism. The reported effective dimension resolved $d_\mathrm{eff}$ is the lowest effective dimension across all three seeds.
\newpara
The average head energies closely reproduce the exact ground energy $E_0$, with maximum per-head quantum variance of the full Hamiltonian at the final iterate across all heads in the order of $10^{-3}$, confirming that the states are nearly perfect eigenstates. Fidelities exceed $0.999$ uniformly across all heads, demonstrating accurate restriction to the ground manifold. The projection matrices have full rank with well-conditioned singular spectra, and the pairwise overlap norms remain well below $0.01$. Taken together, these indicate that for these smaller systems the ensemble not only reproduces each ground state individually, but also spans the degenerate subspace in a numerically stable manner.
\newpara
For even $N \geq 10$, convergence was relatively hard to achieve for both heads, likely arising from a combination of factors. For example, one observed effect is the penalty schedule for encouraging head diversification can interact with network capacity: if enforced too strongly or too early, heads may stabilise into nearly redundant solutions rather than exploring orthogonal directions. Conversely, if imposed too weakly, they may fail to separate before convergence to the ground manifold. Finally, the effective expressive power of the network architecture may begin to saturate at $N \geq 10$. Although this is less likely (as we have achieved a few good convergence simulations for $N = 10$ but are not shown due to being difficult to reproduce across seeds), it remains a relevant concern. In the ST-MH NQS ensemble, all target states share a trunk. Therefore, if a simple architecture converges in the MT-MH case, this does not mean that the same must be true if used in the ST-MH approach. As the purpose of this numerical verification is to only provide a proof-of-principle on the applicability of the ST-MH NQS ensemble, we do not further attempt to fully verify specific reasons for this behaviour observed in this physical model.
\newpara
Nevertheless, we do conduct a very simple ablation study to demonstrate the interplay between the network capacity and the simulation parameters. Table \ref{tab:groundspace-metrics} shows a total of three entries for $N = 4$, two of which are labelled by (B) and (C) respectively. As shown in Table \ref{tab:syssetup} in Appendix \ref{app:syssetup}, the simulation 4(C) was conducted with a trunk width of only 4 (compared to standard trunk width of choice of 32 for the data shown for $N = 4$ in Table \ref{tab:groundspace-metrics}). Further, for $N = 4$(B), the trunk width is only 2. Despite that, Table \ref{tab:groundspace-metrics} shows good convergence for both cases. While this shows the robustness of the NQS ansatz in general and the ST-MH NQS ensemble specifically to parametrise more than one solution with a relatively small network, the change of simulation parameters required for this (shown in Table \ref{tab:syssetup}) highlights the delicate interplay between the network capacity and the simulation parameters. Hence, this further solidifies the argument presented for the difficulty of convergence for $N \geq 10$. 
\newpara
Note that for this model, the two momentum ground states have flat modulus on the common support, therefore $\mathcal{R}_G = \mathrm{span}\lbrace \mathbf{1}_\suppPsi \rbrace$ and hence $r_G = 1$. For the phases, they coincide on $\mathrm{supp}\Phi_A \setminus \mathrm{supp}\Phi_B$ and differ by $\pi$ on $\mathrm{supp}\Phi_B \setminus \mathrm{supp}\Phi_A$. Since $\mathbf{1}_{\mathrm{supp}\Phi_B}$ is independent of $\mathbf{1}_\suppPsi$, this gives $\dim \mathrm{span}\lbrace \mathbf{1}_\suppPsi, \Omega_+, \Omega_-\rbrace = 3$ and hence $r_\mathrm{both} = 3$ and $h_\mathrm{both}^\star = 2$. The ablation results, specifically the case of 4(B) shown in the Table \ref{tab:groundspace-metrics} above, empirically supports the theorem's assertion that with a minimal width of $h = 2$, using the ST-MH NQS ensemble one can indeed resolve the entire degenerate ground manifold.


\subsection{\label{sec:cost}Cost, efficiency, and when to prefer ST-MH}

Both ST-MH and MT-MH ensembles converge to the same set of $K$ degenerate eigenstates, provided the network is expressive enough. The practical question is therefore which architecture reaches that goal with less runtime and memory. In what follows, we provide a \emph{qualitative} computational cost analysis for a simple network architecture and outline a threshold for the efficiency of the ST-MH approach.


\subsubsection{\label{sec:costandparamscount}Costs and parameter count}
Consider the following \emph{qualitative} cost model. Assume that the input to any network considered is $N$ (e.g. the number of sites in the lattice) and that the trunk is composed of two layers and has a final output features vector of size $h_{(s/m)}$. Here, the $(s/m)$ denotes whether we are considering the width for the ST-MH (denoted $s$) or the MT-MH (denoted $m$) trunk. For clarity, we only count matrix-vector multiplications and absorb activation functions and constant factors from automatic differentiation libraries in prefactors. For one forward pass, the floating point operations (FLOPs) count is
\begin{equation}
    F_T^{(s/m)} = (Nh_{(s/m)} + h^2_{(s/m)}) \quad\text{multiplications},
\end{equation}
arising from multiplying an $N \times h_{(s/m)}$ weight matrix by an $N$-vector ($Nh_{(s/m)}$ multiplications, input to first hidden layer), then multiplying a $h_{(s/m)} \times h_{(s/m)}$ weight matrix by the latent vector ($h^2_{(s/m)}$ multiplications, hidden to hidden). In principle, there would be one more last (linear readout) operation, consisting of multiplying an $h_{(s/m)} \times 1$ weight vector by the $h_{(s/m)}$-dimensional latent vector ($h_{(s/m)}$ multiplications, hidden to scalar). This is excluded from the trunk forward pass, as shown above, and is accounted for later in the head FLOPs count below.
\newpara
A backward pass (reverse-mode automatic differentiation) through the network costs almost the same order of magnitude. In principle, we replay every operation in the forward pass and additionally compute a matrix-vector product (vector-Jacobian product) per weight matrix to form gradients. Empirically, the rule of thumb is that it may amount to roughly $2\times$ the cost of a forward pass. Therefore, we can set
\begin{equation}
    B_T^{(s/m)} \simeq 2 F_T^{(s/m)}.
\end{equation}
The \emph{head} (a single linear readout that produces two scalars for phase and amplitude) costs
\begin{equation}
    F_H^{(s/m)} = 2h_{(s/m)}, \qquad B_H^{(s/m)} \simeq 2F_H^{(s/m)}.
\end{equation}
For the ST-MH ensemble, for $N_{MC}$ Monte Carlo samples in one gradient update, the compute cost is
\begin{equation}
\label{eq:coststmh}
    C_{\mathrm{ST-MH}} = N_{MC}(F_T^{(s)} + B_T^{(s)} + K(F_H^{(s)} + B_H^{(s)})) \approx N_{MC}(3F_T^{(s)} + 6Kh_{(s)}).
\end{equation}
Since $F_T^{(s/m)} \propto Nh_{(s/m)} + h^2_{(s/m)}$ while the head term scales like $Kh_{(s/m)}$, the trunk dominates when $3F_T^{(s/m)} \gg 6Kh_{(s/m)} \Leftrightarrow Nh_{(s/m)} + h^2_{(s/m)} \gg 2Kh_{(s/m)}$ (i.e. $K \ll (N + h_{(s/m)})/2$). In this regime, 
\begin{equation}
\label{eq:costmtmh}
    C_{\mathrm{ST-MH}} \approx 3N_{MC} F_T^{(s)}.
\end{equation}
For a MT-MH ensemble, the trunk is duplicated $K$ times, so
\begin{equation}
    C_{\mathrm{MT-MH}} \approx 3KN_{MC}F_T^{(m)}.
\end{equation}
\emph{If} $h_{(s)} = h_{(m)}$, this results in a $K$-fold increase in both compute time and memory footprint as can roughly be seen in Figures \ref{fig:params_vs_k} and \ref{fig:runtime_vs_k}. The number of trainable parameters scales in the same manner. Namely
\begin{align}
\label{eq:paramscountth}
    P_{\mathrm{ST-MH}} &\propto (N+1)h_{(s)} + h^2_{(s)} + 2Kh_{(s)},
    \\
\label{eq:paramscountth2}
    P_{\mathrm{MT-MH}} &\propto K ((N+1)h_{(m)} + h^2{(m)} + 2h_{(m)}).
\end{align}
Therefore, \emph{if} $h_{(s)} = h_{(m)}$, then $P_{\mathrm{MT-MH}} \approx K P_{\mathrm{ST-MH}}$.
\newpara
In \emph{both} ensembles, evaluating the overlap penalty $\mathcal{P}$ and its gradients requires all pairwise normalised overlaps $\sigma_{kl}$, which introduces an additional head-only $O(K^2)$ term. A simple estimate is
\begin{equation}
    C_{\mathrm{pen.}}^{(s/m)} \approx cN_{MC} K(K-1),
\end{equation}
where $c$ collects small constants from the importance sampling estimators in Appendix \ref{app:methodology}. This term is common to both ST-MH and MT-MH, and is typically subdominant when trunks dominate cost. However, for large $K$ and/or large $N_{MC}$ it can become comparable to the trunk cost and should be accounted for in wall-time predictions.


\subsubsection{Efficiency threshold}
ST-MH ensembles can \emph{represent} $K$ orthogonal states with a single latent dimension $h_{(s)}$ as long as $h_{(s)} + 1 \geq r_\mathrm{both}$\footnote{Theorem \ref{thm:sufficiency}'s assertion is a representability statement on a finite support. It does not claim any expressivity, learnability or parameter efficiency assumptions (see Remark \ref{rem:remark3})}, where $r_\mathrm{both}$ is the combined linear rank of the states' log-moduli and phases on the common finite support (see Appendix  \ref{app:sufficiencytheorem}). Because $r_G \le \min\{D+1,|S|\}$ and, for simultaneous moduli and phases, $r_{\mathrm{both}} \le \min\{2D+1,|S|\}$, choosing $h_{(s)} \ge r_{\mathrm{both}}-1$ ensures exact \emph{representability} without unnecessarily inflating the trunk. Since memory co-scales under our simple model with the compute time, ST-MH ensembles are strictly preferable unless one deliberately wishes to eliminate trunk parameter sharing (e.g. to avoid latent space interference).
\newpara
Note that $K$ can exceed $N$. In such cases, if $K \gg N$ the required rank $r_\mathrm{both}$ may also grow with $K$, meaning that ST-MH trunks must be correspondingly wider. This means that while ST-MH remains representationally sufficient whenever $r_\mathrm{both} \leq h + 1$, the efficiency advantage over MT-MH starts to diminish when degeneracy scales \emph{extensively} compared to $N$. However, to obtain an empirical threshold, consider once more the costs in equations \eqref{eq:coststmh} and \eqref{eq:costmtmh} and further include the dropped head terms $6Kh_{(s/m)}$. Now assume that the MT-MH ansatz can represent each state with a trunk width $h_{(m)}$ and further, $h_{(s)} \neq h_{(m)}$. Now, solving
\begin{align}
    C_\mathrm{ST-MH} & \leq C_\mathrm{MT-MH} 
    \\
    Nh_{(s)} + h^2_{(s)} + 2Kh_{(s)} & \leq K(Nh_{(m)} + h^2_{(m)}) + 2Kh_{(m)},
\end{align}
for $h_{(s)}$ one can obtain a precise threshold
\begin{equation}
\label{eq:threshold}
    h_{(s)}^\star = \frac{-(N + 2K) + \sqrt{(N + 2K)^2 + 4K(Nh_{(m)} + h_{(m)}^2 + 2h_{(m)})}}{2}.
\end{equation}
This shows that if $h_{(s)} \in [0, h_{(s)}^\star]$, the ST-MH approach is cheaper or equal to the MT-MH. Otherwise, the MT-MH is computationally cheaper. Precisely, one can write a slowdown factor
\begin{equation}
    R(h_{(s)}) = \frac{C_\mathrm{ST-MH}(h_{(s)})}{C_\mathrm{MT-MH}(h_{(m)})},
\end{equation}
where if $R(h_{(s)}) < 1$, the ST-MH ensemble is faster by a factor of $1/R(h_{(s)})$ while if $R(h_{(s)}) > 1$, it is slower by a factor of $R(h_{(s)})$. 
\newpara
In this qualitative model, we assume a 2-hidden layer fully connected trunk with equal hidden widths $h_{(s/m)}$. The threshold above isolates the \emph{last layer} feature sizes $h_{(s/m)}$ as the deciding knobs: for a given per-head budget $h_{(m)}$, ST-MH is preferable whenever the shared last layer width $h_{(s)}$ stays below the break-even value in equation \eqref{eq:threshold}. In practice, however, $h_{(s)}$ and $h_{(m)}$ are just the dimensions of the trunks' \emph{final} feature vector. The preceding layers can be deep and wide. For an arbitrary fully connected network, one should replace $F_T$ by the layer-wise sum $\sum_l d_{l-1}d_l$ (and keep the head term) where $d_i$ is the size of the layer $i$ and $d_0 = N$. For different architectures (e.g. convolutional networks), the accounting is architecture specific. A similar substitution uses the per-layer FLOP counts  (e.g. for convolutional networks, kernel height $\times$ kernel width $\times$ in-channels $\times$ out-channels $\times$ spatial positions), with the same head term unchanged.
\newpara
In this simple cost model, the qualitative message survives: efficiency hinges on the compactness of the shared representation, but the numerical threshold shifts with architectural choices upstream of the last layer. Therefore, while the equation says the sizes $h_{(s/m)}$  decide, in reality whether a given architecture achieves a small effective $h_{(s)}$ without inflating earlier layers is ultimately an empirical question which our simple cost model only approximates. Note that very large $h$ can make the Monte Carlo estimator of the natural gradient matrix noisy, and may require a larger amount of samples $N_{MC}$ or stronger regularisers. In that regime, the theoretical advantage may shrink.


\section{\label{sec:conclusion}Conclusion}

This work presents a unified and systematic approach for training NQS ensembles to capture degenerate ground states using variational Monte Carlo. By formulating exact gradient expressions and leveraging a linear rank based expressivity theorem, we demonstrated that the single-trunk multi-head (ST-MH) ensemble can represent all degenerate states exactly whenever the trunk width satisfies $h + 1 \geq r_\mathrm{both}$, where $r_\mathrm{both}$ is the combined linear rank of the states' log-moduli and phases on a common support in the degenerate target manifold where all states are non-vanishing.
\newpara
When this criterion is met, ST-MH can offer a, in some situations, substantial computational and memory advantage over the multi-trunk multi-head (MT-MH) ensemble, with no loss in accuracy. Our numerical experiments for the frustrated spin-$\tfrac{1}{2}$ $J_1-J_2$ Heisenberg chain on a periodic ring with even number of sites, considered at the Majumdar-Ghosh point, confirm the theoretical cost advantage expectations. The conducted full degenerate ground space resolution, including fidelity tests, orthogonality and complete spanning of the true degenerate ground manifold further strengthen the applicability of this ensemble ansatz, all the meanwhile reducing runtime and the required memory resources relative to an MT-MH ensemble with the same trunk width per-head would require. Through ablation studies, we provide empirical support for the theorem's assertion that the minimal trunk latent width $h^\star_\mathrm{both} = 2$ (for the considered model) needed to resolve the full degenerate ground state manifold is indeed true by converging to the two momentum eigenstates with a trunk of width $h = 2$.
\newpara
A qualitative computational cost analysis for a simple 2-hidden layer fully connected network shows that one can (in theory) provide a threshold on the ST-MH trunk width whereby the ST-MH approach offers significant increase in efficiency compared to the MT-MH approach. These findings establish practical and scalable guidelines for selecting NQS architectures in multi-state quantum learning tasks, highlighting ST-MH as potentially a preferred approach in settings where expressivity conditions are met.
\newpara
Moreover, while the ST-MH approach further extends the expressivity of the NQS ansatz by showing it is robust enough to represent many states with one network, this architectural separation between shared features and linear heads may also prove valuable beyond single-system degenerate eigenspaces, such as in transfer learning situations using foundation neural networks \cite{Rende2025} where a single trunk may be able to learn general representations adaptable across multiple quantum systems.


\section{\label{sec:ack}Acknowledgments}
The author thanks Hanno Sahlmann for the helpful comments and conversations which helped bring this work to its current form. The author gratefully acknowledges the scientific support and HPC resources provided by the Erlangen National High Performance Computing Center (NHR@FAU) of the Friedrich-Alexander-Universität Erlangen-Nürnberg (FAU). The hardware is funded by the German Research Foundation (DFG).

\appendix


\clearpage


\section*{Appendices}

\renewcommand{\thesection}{\Alph{section}}

\setcounter{footnote}{0}


\section{\label{app:sufficiencytheorem}Representability theorem}

Let $\mathcal{C}$ be a finite configurations (e.g. spin configurations) set of size $M$. It is convenient to identify real-valued functions on $\mathcal{C}$ with vectors in $\mathbb{R}^M$ by fixing an ordering $\mathcal{C} : = \lbrace x_1, \cdots, x_M \rbrace$ and identifying real functions $A$ on $\mathcal{C}$ with their evaluation (column) vectors such that
\begin{equation}
    A : \mathcal{C} \rightarrow \mathbb{R} \Leftrightarrow \ev{\mathcal{C}}{A} := (A(x_1), \cdots, A(x_M))^\top \in \mathbb{R}^{M \times 1},
\end{equation}
which is simply a linear isomorphism. Now suppose we are given $D$-many target complex-valued wave-functions $\lbrace \Psi^{(j)} : \mathcal{C} \rightarrow \mathbb{C}\rbrace_{j = 1}^D$. For any $x \in \mathcal{C}$, and when $\Psi^{(j)} \neq 0$, one can write 
\begin{equation}
    \Psi^{(j)}(x) = \expfunc{G_{\Psi^{(j)}}(x) + \ii\Omega_{\Psi^{(j)}}(x)},
\end{equation}
where \begin{equation}
    G_{\Psi^{(j)}}(x) := \ln|\Psi^{(j)}(x)| \farcomma \Omega_{\Psi^{(j)}}(x) := \arg\Psi^{(j)}(x).
\end{equation}
When $\Psi^{(j)}(x) = 0$, then the corresponding $G_{\Psi^{(j)}}(x)$ is undefined. Therefore, let
\begin{equation}
    \mathcal{C} \supseteq \mathscr{S} := \bigcap_{j = 1}^D \mathrm{supp} \Psi^{(j)} \farcomma \mathrm{supp}\Psi^{(j)} := \lbrace x \in \mathcal{C} \,:\, |\Psi^{(j)}(x)| > 0\rbrace,
\end{equation}
be denoted the \emph{common support}. Note that from this point onward, $\mathscr{S}$ \emph{is assumed to have a fixed ordering} and hence evaluation vectors can be identified to functions on $\mathscr{S}$. Configurations with $\Psi(x) = 0$ are outside the domain of $G_\Psi$ and are not representable exactly by this exponential ansatz. A $\epsilon$-regularisation yields approximation but not identity at such points. From here onward, all statements in this appendix concern pointwise equality on $\suppPsi$.
\newpara
Note that the target log-moduli and phases $G_{\Psi^{(j)}}$ and $\Omega_{\Psi^{(j)}}$ can therefore be associated with their evaluation vectors $\ev{\mathscr{S}}{G_{\Psi^{(j)}}}$ and $\ev{\mathscr{S}}{\Omega_{\Psi^{(j)}}}$, respectively, on the common support. To keep the notation compact, we will often interchange the two.

\begin{definition}
\label{def:rank}
    Let $\suppPsi := \bigcap_{j = 1}^D \mathrm{supp}\Psi^{(j)}$ and fix single-valued phase branches $\{\Omega_{\Psi^{(j)}}\}_{j=1}^D$. Define the linear modulus and phase spans
    \begin{align}
    \mathcal{R}_G & := \mathrm{span} \left\lbrace \mathbf{1}_\mathscr{S} , \ev{\mathscr{S}}{G_{\Psi^{(1)}}}, \cdots, \ev{\mathscr{S}}{G_{\Psi^{(D)}}}\right\rbrace \subset \mathbb{R}^{|\suppPsi|} ,
    \\
\mathcal{R}_\Omega & := \mathrm{span}\lbrace\mathbf{1}_\mathscr{S}, \ev{\mathscr{S}}{\Omega_{\Psi^{(1)}}},\ldots,\ev{\mathscr{S}}{\Omega_{\Psi^{(D)}}}\rbrace\subset \mathbb{R}^{|\suppPsi|},
\end{align}
where $\mathbf{1}_\mathscr{S}(x) = 1$, and define
\begin{equation}
\mathcal{R}_{\mathrm{both}} := \mathrm{span}\bigl(\mathcal{R}_G \cup \mathcal{R}_\Omega\bigr).
\end{equation}
    Write $r_G=\dim \mathcal{R}_G$, $r_\Omega=\dim \mathcal{R}_\Omega$ as the linear modulus and linear phase ranks respectively, and $r_{\mathrm{both}}=\dim \mathcal{R}_{\mathrm{both}}$.
    \\
\end{definition}
\noindent
Recall that for the ST-MH ansatz, one has one differentiable trunk $f_\vartheta : \mathscr{S} \rightarrow \mathbb{R}^{h \times 1}$ which can be written component-wise as $(f_{\vartheta, 1}(x), \cdots, f_{\vartheta, h}(x))^\top$. Here, $f_{\vartheta, i} : \mathscr{S} \rightarrow \mathbb{R}$ is scalar at each $x \in \mathscr{S}$ and its evaluation vector $\ev{\mathscr{S}}{f_{\vartheta, i}}$ lies in $\mathbb{R}^{|\mathscr{S}|}$. For head $k$ ($k = 1, \cdots, K$), real vectors $\alpha_k, \varphi_k \in \mathbb{R}^{1\times h}$ and scalars $\beta_k, \gamma_k \in \mathbb{R}$, define the complex vectors
\begin{equation}
\label{eqapp:chiandc}
    \chi_k := \alpha_k + \ii\varphi_k \in \mathbb{C}^{1 \times h} \farcomma c_k := \beta_k + \ii \gamma_k \in \mathbb{C}.
\end{equation}
The single-trunk multi-head wave-function of head $k$ is then, for all $x \in \mathscr{S}$, given by
\begin{equation}
\label{eqnapp:psik}
    \psi_k(x) = \expfunc{\chi_k f_\vartheta(x) + c_k}.
\end{equation}
\begin{lemma}\label{lemma:affrankbound}(Affine rank bound)
    Fix $\vartheta$. Now consider the real vectors, denoted the realised log-moduli, $g_k := (\ln|\psi_k(x_1)|, \cdots, \ln|\psi_k(x_{|\mathscr{S}|})|)^\top \in \mathbb{R}^{|\mathscr{S}| \times 1}$ for $k = 1, \cdots, K$. The affine subspace $\mathcal{A}:= \mathrm{aff}\left\lbrace g_1, \cdots, g_K\right\rbrace \subset \mathbb{R}^{|\mathscr{S}| \times 1}$ has affine dimension at most $h+1$.
\end{lemma}
\begin{proof}
    For each $x \in \mathscr{S}$, then
    \begin{equation}
    \label{eqapp:gk}
        g_k(x) = \Re \left[ \chi_k f_\vartheta(x) + c_k\right] = \alpha_k f_\vartheta(x) + \beta_k.
    \end{equation}
    Let $\mathrm{F} \in \mathbb{R}^{h \times |\mathscr{S}|}$ be a matrix whose $(\mu, i)$-entry is $\mathrm{F}_{\mu i} := f_{\vartheta, \mu}(x_i)$. Equation \eqref{eqapp:gk} for all configurations is then
    \begin{equation}
    \label{eqnapp:gk2}
        g_k = \mathrm{F}^\top \alpha_k^\top + \beta_k \mathbf{1}_{\mathscr{S}}.
    \end{equation}
    Choose $g_1$ as origin of the affine space. For $k \geq 2$, then
    \begin{equation}
        g_k - g_1 = \mathrm{F}^\top (\alpha_k - \alpha_1)^\top + (\beta_k - \beta_1)\mathbf{1}_{\mathscr{S}}.
    \end{equation}
    Because columns of $\mathrm{F}^\top$ live in an $h$-dimensional linear subspace of $\mathbb{R}^{|\mathscr{S}|}$ and $\mathbf{1}_{\mathscr{S}}$ adds at most one more dimension, the $\mathrm{span}\lbrace g_k - g_1 \,|\, k = 2, \cdots , K\rbrace$ lies in a linear subspace $L(f) := \mathrm{im}\mathrm{F}^\top + \mathrm{span}\lbrace\mathbf{1}_{\mathscr{S}}\rbrace$ whose dimension $ \dim L(f)\leq h +1$ and hence, $\dim\mathcal{A} \leq h + 1$.

\end{proof}

\begin{corollary}\label{cor:phaseaffinerank}(Affine rank bound for realised phase lifts)
    Fix $\vartheta$. For head $k$, write the realised (unwrapped) phase lift
    \begin{equation}
        \omega_k(x) := \Im[ \chi_k f_\vartheta(x) + c_k] = \varphi_k f_\vartheta(x) + \gamma_k,
    \end{equation}
    for all $x \in \mathscr{S}$. Then $\dim\mathrm{aff} \lbrace \omega_1, \cdots, \omega_K \rbrace \leq h + 1$ where $\omega_i$ now are represented by their evaluation vectors.
\end{corollary}
\begin{proof}
    Exactly as in Lemma \ref{lemma:affrankbound}, for $\mathrm{F} \in \mathbb{R}^{h \times |\mathscr{S}|}$ with $\mathrm{F}_{\mu i} := f_{\vartheta, \mu}(x_i)$, we have $\mathrm{F}^\top \varphi^\top_k + \gamma_k \mathbf{1}_\mathscr{S}$. Hence $\omega_k - \omega_1 = \mathrm{F}^\top (\varphi_k - \varphi_1)^\top + (\gamma_k - \gamma_1) \mathbf{1}_\mathscr{S}$. Thus, $\mathrm{span}\lbrace \omega_k - \omega_1 \,|\, k = 2, \cdots, K \rbrace \subseteq \mathrm{im} \mathrm{F}^\top + \mathrm{span}\lbrace \mathbf{1}_\mathscr{S}\rbrace$, whose dimension is at most $h+1$.
\end{proof}

\begin{remark}
    If one records principal-branch phases $\arg\psi_k \in (-\pi, \pi]^{|\mathscr{S}|}$, then there exists integer vectors $n_k \in \mathbb{Z}^{|\mathscr{S}|}$ with $\arg \psi_k = \omega_k - 2\pi n_k$. The family $\lbrace \arg\psi_k \rbrace$ need not lie in a single affine subspace because different $n_k$ may occur. For capacity and for matching chosen target branches on $\mathscr{S}$, we work with the lifts $\omega_k$ (or fix single-valued branches), for which the affine bound above holds and which satisfy $e^{\ii\omega_k(x)} = \psi_k(x) / |\psi_k(x)|$ pointwise.
\end{remark}
\noindent
The Lemma and corollary above are then a structural capacity statement for a \emph{fixed} trunk. Namely, the set of all realised log-moduli and phases lie in an affine subspace of dimension at most $h + 1$ and thus they quantify how large a family of log-moduli and phases can a trunk of fixed $h$ represent, irrespective of the targets.

\begin{theorem}[ST-MH representability and minimal width on the common support]\label{thm:sufficiency}
Let the set $\lbrace\Psi^{(j)}(x)\rbrace_{j=1}^D$ be a degenerate manifold, and let $\suppPsi$ be the common support. Fix single-valued phase branches $\lbrace\Omega_{\Psi^{(j)}}\rbrace_{j=1}^D$ on $\suppPsi$, and let $r_{\mathrm{both}}=\dim \mathcal R_{\mathrm{both}}$ as in Definition \ref{def:rank}. The following are equivalent:
\begin{enumerate}
    \item There exists a single trunk $f_\vartheta:\suppPsi \rightarrow \mathbb{R}^h$ and linear heads such that $\psi_k(x)=\Psi^{(k)}(x)$ for all $k \le D$ and all $x \in \suppPsi$.
    
    \item $h \ge h_{\mathrm{both}}^\star = r_{\mathrm{both}}-1$ (equivalently, $r_{\mathrm{both}}\le h+1$).
\end{enumerate}
In particular, the minimal width to realise both the log-moduli and the phases with the same trunk on $\suppPsi$ is $h_{\mathrm{both}}^\star$.
\end{theorem}
\begin{proof}
    Let $\mathcal{R}_{\mathrm{both}}$ be as defined in Definition \ref{def:rank}. Since $\mathcal{R}_{\mathrm{both}} \subset \mathbb{R}^{|\mathscr{S}|}$ and $\mathbf{1}_\mathscr{S} \in \mathcal{R}_{\mathrm{both}}$, there exists a basis $\lbrace\mathbf{1}_\mathscr{S}, b_1, \cdots, b_{r_\mathrm{both} - 1}\rbrace$ of $\mathcal{R}_\mathrm{both}$. Now assume $r_\mathrm{both} \leq h + 1$ and define the following $(h+1)$-many vectors
    \begin{equation}
        u_0 := \mathbf{1}_\mathscr{S} \farcomma u_i := b_i \quad(i = 1, \cdots, r_\mathrm{both} - 1).
    \end{equation}
    where if $h+1 > r_\mathrm{both}$, pad with any additional vectors $u_{r_\mathrm{both}}, \cdots, u_h$ from $\mathcal{R}_\mathrm{both}$. Therefore, by construction $\mathrm{span}\lbrace u_0, \cdots, u_h\rbrace = \mathcal{R}_\mathrm{both}$, since independence is not required for a spanning set. It is always possible to choose such a spanning set which contains the constant column $\mathbf{1}_\mathscr{S}$.
    \newpara
    Since $u_i$ are vectors, that means that there exists real-valued functions $U_i : \mathscr{S} \rightarrow \mathbb{R}$ to which their evaluation vector is precisely $u_i$, for all $i$. Therefore, we can define the trunk on $\mathscr{S}$ by setting $f_{\vartheta, i} := U_i$ where $i = 1, \cdots, h$ and hence the evaluation vector of $f_{\vartheta, i}$ is equal to $u_i$.  Now write a matrix $X = \left[\mathbf{1}_\mathscr{S} \,\,\, f_{\vartheta, 1} \,\,\,  \cdots \,\,\, f_{\vartheta, h}\right] \in \mathbb{R}^{|\suppPsi| \times (h+1)}$. Then, $\mathrm{col}(X) = \mathrm{span}\lbrace u_0, \cdots, u_h\rbrace = \mathcal{R}_\mathrm{both}$. Consequently, since every target evaluation vector $G_{\Psi^{(j)}}$ and $\Omega_{\Psi^{(j)}}$ lies $\mathcal{R}_\mathrm{both}$, then there exist $\theta_j = (\beta_j, \alpha_j^\top)^\top$ and $\eta_j = (\gamma_j, \varphi_j^\top)^\top$ with
    \begin{equation}
        X \theta_j = G_{\Psi^{(j)}}  \farcomma  X \eta_j \equiv \Omega_{\Psi^{(j)}}.
    \end{equation}
    Interpreting these equalities pointwise on $\mathscr{S}$ gives, for all $x\in\mathscr{S}$,
    \begin{equation}
        \alpha_j f_\vartheta(x) + \beta_j = G_{\Psi^{(j)}}(x) \farcomma \varphi_j f_\vartheta(x) + \gamma_j = \Omega_{\Psi^{(j)}}(x) \,(\mathrm{mod} 2\pi)
    \end{equation}
    Since $\alpha_j f_\vartheta(x) + \beta_j = \ln|\psi_j(x)|$ and $\varphi_j f_\vartheta(x) + \gamma_j = \arg\psi_j(x)$, then 
    \begin{equation}
        \psi_j(x) = \expfunc{G_{\Psi^{(j)}}(x) + \ii\Omega_{\Psi^{(j)}}(x)} = \Psi^{(j)}(x)\quad (x\in \suppPsi),
    \end{equation}
    for all $x \in \mathscr{S}$.
    \newpara
    Lastly, assume, for contradiction, that $r_{\mathrm{both}} > h+1$ and yet there exists a trunk $f_\vartheta$ and heads with $\psi_k = \Psi^{(k)}$ on $\suppPsi$ for all $k$. Then both the realised log-moduli and the chosen phase branches lie in $\mathrm{col}(X)$, implying $\dim\,\mathrm{col}(X) \ge r_{\mathrm{both}}$. However, $\dim\,\mathrm{col}(X) \le h+1$, causing a contradiction. Therefore, no single trunk of width $h$ can represent all $D$ states on $\suppPsi$ when $r_{\mathrm{both}} > h+1$. We denote the minimal width required for exact representation to be $h^\star_\mathrm{both} = r_\mathrm{both} - 1$.
\end{proof}
\begin{corollary}
    Let $T \subseteq \suppPsi$ be any non-empty subset (e.g. a symmetry or conservation law sector). Since evaluation vectors can be defined on $T_1$, define the subspaces $\mathcal{R}_G(T), \mathcal{R}_\Omega(T)$ and $\mathcal{R}_\mathrm{both}(T)$ now with the evaluation vectors restricted on $T$. Consequently, $r_\mathrm{both}(T) \leq r_\mathrm{both}(\suppPsi)$ since the restriction is a linear map that cannot increase the linear span generated by the (restricted) columns and hence, the dimensions cannot increase. Thus, a similar sector-specific representability assertion (following the same steps of Theorem \ref{thm:sufficiency}) can be done, yielding now representability criterion of $h_\mathrm{both}^\star(T) = r_\mathrm{both}(T) - 1$.
\end{corollary}
\begin{remark}
    In general, for a fixed trunk, the realised log-moduli and phases lie in the linear subspace $\mathrm{span}\lbrace \mathbf{1}_\mathscr{S}, f_{\vartheta, 1}, \cdots, f_{\vartheta, h}\rbrace$. If the features are all zero, all equal, or some are constant, this span shrinks and the affine dimension bound in Lemma \ref{lemma:affrankbound} becomes even tighter. This degeneracy consequently only makes the representation harder. In the proof above, we avoid this by choosing features whose evaluation columns form (together with $\mathbf{1}_\mathscr{S}$) a spanning set of $\mathcal{R}_\mathrm{both}$. When $h + 1 > r_\mathrm{both}$, we may pad with arbitrary extra features (including zeros). These do not enlarge the column space and do not affect solvability of the linear system.
\end{remark}
\begin{remark}
\label{rem:remark3}
     Theorem \ref{thm:sufficiency} is a representability result on the finite common support. It does not guarantee that any particular neural architecture of width $h$ can realise the target functions or that optimisation will find them. Such statements require extra expressivity assumptions. The proof's assignment of $f_{\vartheta, i} = u_i$ is an existence claim which for standard (e.g. MLP) trunks is implementable exactly on $\suppPsi$. There is no claim made regarding learnability or parameter efficiency.
\end{remark}
\begin{remark}
    Beyond the rank condition $r_\mathrm{both} \leq h + 1$, exact representability also requires that, after choosing a single-valued branch for each phase function $\Omega_{\Psi^{(j)}}$, the phases lie in the same affine span generated by the trunk features and a constant (mod $2\pi$). This is an expressivity requirement on $f_\vartheta$. On a finite configuration space $\mathscr{S}$, a global branch can always be selected. In models with twisted boundary conditions or magnetic flux, the practical question is whether those phases are representable as affine functions of $f_\vartheta$, not a topological obstruction.
\end{remark}
\noindent
Therefore, it is shown that using the ST-MH ensemble, one can match all $D$ eigenstates exactly on $\suppPsi$ if $r_\mathrm{both} \leq h + 1$ where $r_\mathrm{both}$ is computed from the chosen single-valued phase branches on $\suppPsi$. Conversely, if $r_\mathrm{both} > h + 1$, no single-trunk width $h$ can reproduce all $D$ eigenstates. This is independent of the head count $K$ (as long as $K \geq D$).
\newpara
This failure is independent of the head count, it is simply a feature space bottleneck. Namely, all heads share the same feature vector and each head can only take linear combinations of those $h + 1$-many numbers. If $r > h + 1$, the set of eigenstate moduli need more than $h + 1$ linearly independent real functions to be written and no amount of extra heads can create new features, they only reuse the same $h + 1$ coordinates.


\section{\label{app:VMCNQSClasses}Variational Monte Carlo with penalty term}

For any normalisable $\psi$, define the Born probability $p_\psi(x)$ as
\begin{equation}
    N[\psi]:=\sum_{x \in \configSpace} |\psi(x)|^2 \farcomma p_\psi(x) = \frac{|\psi(x)|^2}{N[\psi]},
\end{equation}
and the local (energy) estimator takes the standard form \cite{Lange:2024nsr,Hibat_Allah_2020}
\begin{equation}
    E_{loc, \psi}(x) = \sum_{y \in \configSpace} \op{H}_{xy} \frac{\psi(y)}{\psi(x)}.
\end{equation}
For some (Hermitian) Hamiltonian $\op{H}$. Expectation values with respect to $p_\psi(x)$ are then \cite{Lange:2024nsr}
\begin{equation}
    \expect{A}_{p_\psi} = \sum_{x \in \configSpace} p_\psi(x) A(x).
\end{equation}
For the ST-MH, the wave functions share a trunk $f_\vartheta : \configSpace \rightarrow \mathbb{R}^{h \times 1}$ with $\vartheta \in \mathbb{R}^{P_T}$. For each head $k = 1, \cdots, K$ choose real vectors and scalars $\alpha_k , \varphi_k \in \mathbb{R}^{1 \times h}$ and $\beta_k, \gamma_k \in \mathbb{R}$ and define the complex coefficients $\chi_k := \alpha_k + \ii \varphi_k \in \mathbb{C}^{1 \times h}$ and $c_k := \beta_k + \ii \gamma_k \in \mathbb{C}$. The feature vector at configuration $x$ is then $ f_\vartheta(x) = (f_{\vartheta, 1}(x), \cdots, f_{\vartheta, h}(x))^\top \in \mathbb{R}^{h \times 1}$. Note that in the complex-parameter case, whenever a derivative acts on $\psi_k^*$, the corresponding log-derivative must be complex conjugated. We therefore adopt this convention below.
\newpara
Given the above, the head-$k$ wave-function then takes the form
\begin{equation}
    \psi_k(x) = \expfunc{\chi_k f_\vartheta(x) + c_k}.
\end{equation}
Consider now the head parameters. For component $\alpha_{k_\mu}$ with $\mu = 1, \cdots, h$, the log derivative is
\begin{equation}
    \partial_{\alpha_{k_\mu}}\ln\psi_k(x) = f_{\vartheta, \mu} (x).
\end{equation}
Similarly, one can compute the remaining log derivatives as $\partial_{\beta_{k}}\ln\psi_k(x) = 1$, $\partial_{\varphi_{k_\mu}}\ln\psi_k(x) = \ii f_{\vartheta, \mu}(x)$ and $\partial_{\gamma_{k}}\ln\psi_k(x) = \ii$. For the trunk parameters, one writes the Jacobian
\begin{equation}
    \partial_{\vartheta_i} f_\vartheta(x) = (\partial_{\vartheta_i} f_{\vartheta, 1}(x), \cdots, \partial_{\vartheta_i} f_{\vartheta, h})^\top.
\end{equation}
All together, the log derivative vectors take the form
\begin{align}
    \mathcal{O}_k := \nabla_{\theta_k^{(h)}} \ln \psi_k(x) \in \mathbb{C}^{P_H},
    \\
    (\mathcal{O}_k^i)^\top (x) := \partial_{\vartheta_i} \ln\psi_k(x) \in \mathbb{C}.
\end{align}
The energy expectation of head-$k$ is straight-forward. Namely,
\begin{equation}
    E_k := \frac{\innerProduct{\psi_k}{\op{H}\psi_k}}{N[\psi_k]} = \sum_{x \in \configSpace}p_{\psi_k}(x) E_{loc, \psi_k}(x).
\end{equation}
To compute the gradient of $E_k$, let $P$ be any arbitrary head parameter. Then
\begin{equation}
    \partial_P E_k = \partial_P \sum_{x \in \configSpace}p_{\psi_k}(x) E_{loc, \psi_k}(x).
\end{equation}
where
\begin{align}
    \partial_P p_{\psi_k}(x) & = 2\Re\left[ (\mathcal{O}_{k, P}(x) - \expect{\mathcal{O}_{k, P}(x)}_k) p_{\psi_k}(x)\right],
    \\
    \partial_P E_{loc, \psi_k} & = \sum_{y \in \configSpace} \op{H}_{xy}(\mathcal{O}_{k, P}(y) - \mathcal{O}_{k, P}(x))\frac{\psi_k(y)}{\psi_k(x)}.
\end{align}
Taking the expectation and separating the real part gives an identity akin to the one obtained in the case of a single network \cite{Lange:2024nsr,Hibat_Allah_2020}
\begin{equation}
\label{eqapp:energyderhead}
    \partial_P E_k = 2\Re\left[ \expect{(\mathcal{O}_{k, P} - \expect{\mathcal{O}_{k, P}}_{p_{\psi_k}})^*(E_{loc, \psi_k} - E_k)}_{p_{\psi_k}} \right].
\end{equation}
For now a shared trunk parameter $\vartheta_i$, the same algebra applies. One simply replaces $\mathcal{O}_{k, P} \rightarrow (\mathcal{O}_k^i)^\top$. Hence, 
\begin{equation}
\label{eqapp:energydertrunk}
    \partial_{\vartheta_i} E_k = 2\Re\left[ \expect{((\mathcal{O}_k^i)^\top - \expect{(\mathcal{O}_k^i)^\top}_{p_{\psi_k}})^*(E_{loc, \psi_k} - E_k)}_{p_{\psi_k}} \right].
\end{equation}
Where the $x$ dependence is implied. The cost function now also includes an overlap and orthogonality penalty term as shown in equation \eqref{eq:costMTMH} (there for the case of MT-MH NQS ensembles). Define the unnormalised overlaps as $\Sigma_{kl} = \sum_{x \in \configSpace} \psi_k^*(x) \psi_l(x)$. Letting $N_k := \Sigma_{kk}$, then
\begin{equation}
\label{appeq:normoverlapmat}
\sigma_{kl} := \frac{\Sigma_{kl}}{\sqrt{N_k N_l}}.
\end{equation}
Now, for some head parameter $P \in \theta_k^{(h)}$, a straightforward calculation shows that the derivatives take the form
\begin{equation}
    \partial_P \sigma_{kl} = \sqrt{\frac{N_k}{N_l}} \expect{\left(\mathcal{O}_{k, P}(x) - \expect{\mathcal{O}_{k, P}(x)}_{p_{\psi_k}}\right)^* \frac{\psi_l(x)}{\psi_k(x)}}_{p_{\psi_k}}.
\end{equation}
On the other hand, for some trunk parameter $\vartheta_i$, both $\psi_l(x)$ and $\psi_k(x)$ vary, and symmetrising gives
\begin{align}
    \partial_{\vartheta_i} \sigma_{kl} = \sqrt{\frac{N_k}{N_l}}&\expect{\left((\mathcal{O}_k^i)^\top(x) - \expect{(\mathcal{O}_k^i)^\top(x)}_{p_{\psi_k}} \right)^* \frac{\psi_l(x)}{\psi_k(x)}}_{p_{\psi_k}}
    + (k \leftrightarrow l).
\end{align}
Now, for the penalty $\mathcal{P} = \frac{1}{2} \sum_{k \neq l} |\sigma_{kl}|^2$, and for any parameter $P$, then
\begin{equation}
\label{eqapp:penaltyderhead}
    \partial_P \mathcal{P} = \Re\left[\sum_{k \neq l} \sigma^*_{kl} \left(\expect{\mathcal{O}_{k, P}^*}_{kl} - \sigma_{kl} \expect{\mathcal{O}_{k, P}^*}_{kk} \right)\right],
\end{equation}
where 
\begin{equation}
    \expect{\mathcal{O}_{k, P}^*}_{kl} := \sqrt{\frac{N_k}{N_l}} \expect{\mathcal{O}_{k, P}^*(x) \frac{\psi_l(x)}{\psi_k(x)}}_{p_{\psi_k}}
\end{equation}
Thus for the full cost function, given some ensemble weights $w_k$ such that $\sum_k w_k = 1$ and a penalty strength $\lambda$, then for a head parameter $P \in \theta_k^{(h)}$
\begin{equation}
    \partial_P C = w_k \partial_P E_k + \lambda \partial_P \mathcal{P},
\end{equation}
where one uses equations \eqref{eqapp:energyderhead} and \eqref{eqapp:penaltyderhead}. For a shared trunk parameter $\vartheta_i$, then
\begin{equation}
    \partial_{\vartheta_i} C = \sum_{k = 1}^K w_k \partial_{\vartheta_i} E_k + \lambda \partial_{\vartheta_i} \mathcal{P},
\end{equation}
where one uses equation \eqref{eqapp:penaltyderhead} with $\mathcal{O}_{k, P} \rightarrow (\mathcal{O}_k^i)^\top$ and equation \eqref{eqapp:energydertrunk} for the trunk derivatives of the energy.


\section{\label{app:methodology}Optimisation and sampling protocol for ST-MH NQS ensembles}

In what follows, we adopt the notation of Appendix \ref{app:VMCNQSClasses}. Namely for a given configuration $x \in \configSpace$, the ST-MH NQS ensemble reads
\begin{equation}
    \psi_k(x) = \expfunc{\chi_k f_\vartheta(x) + c_k},
\end{equation}
whereby $\chi_k \in \mathbb{C}^{1 \times h}, c_k \in \mathbb{C}$ and $f_\vartheta : \configSpace \rightarrow \mathbb{R}^{h \times 1}$. We further let $\ell_k(x) := \ln \psi_k(x)$.


\subsection{Sampling}

For sampling, we use two complementary modes for estimating the energies and overlaps. The first is \emph{per-head} sampling, where for each head $k$, we run an independent Markov chain targeting $q_k(x) := p_{\psi_k}(x)$ and form simple average over i.i.d. or thinned samples $x_i \sim q_k$. This maximises effective sample size (ESS) per head, but yields head-specific supports. The second mode is a \emph{shared-mixture} sampling (self-normalised importance sampling, SNIS) where a single chain targets the uniform (unnormalised) mixture
\begin{equation}
    q_{\mathrm{mix}}(x) := \frac{1}{K} \sum_{m = 1}^K |\psi_m(x)|^2,
\end{equation}
and self-normalised importance weights for head $k$ are
\begin{equation}
    \omega_{k, i} \propto \frac{|\psi_k(x_i)|^2}{q_{\mathrm{mix}}(x_i)} \farcomma \sum_{i = 1}^S \omega_{k, i} = 1.
\end{equation}
Note that normalising constants cancel both in acceptance ratios and in the normalised weights, thus no explicit per-head factors $N_m$ are needed in the mixture. The per-head ESS is then
\begin{equation}
    \mathrm{ESS}_k = \frac{1}{\sum_{i}\omega_{k, i}^2} \in [1, S].
\end{equation}
This mixture sampling ensures that all heads see the union of supports from all states, especially early in training, which stabilises overlaps and orthogonality pressure.


\subsection{Estimators and cost}

For head $k$, the Rayleigh quotient is estimated via local energies as 
\begin{equation}
\label{eqapp:energyestimator}
    \hat{E}_k = \begin{cases}
        \frac{1}{S} \sum_{i=1}^S E_{loc}^{(k)}(x_i) &\,,\, x_i \sim q_k \, (\text{per-head}) \\
        \sum_{i = 1}^S \omega_{k, i} E_{loc}^{(k)}(x_i) &\,,\, x_i \sim q_{\mathrm{mix}} \, (\text{mixture}).
    \end{cases}
\end{equation}
The normalised overlaps use the ratio form shown in Appendix \ref{app:VMCNQSClasses}. With
\begin{equation}
    L_{kl}(x) := \ell_l(x) - \ell_k(x) = \ln \frac{\psi_l(x)}{\psi_k(x)},
\end{equation}
an estimator of $\sigma_{kl}$ is
\begin{equation}
\label{eqapp:estimatorsigma}
    \hat{\sigma}_{kl} = \begin{cases}
        \frac{1}{S} \sum_{i = 1}^S \expfunc{L_{kl}(x_i)} & \,,\, x_i \sim q_k, \\
        \sum_{i = 1}^S \omega_{k, i}\expfunc{L_{kl}(x_i)} &\,,\, x_i \sim q_{\mathrm{mix}} .
    \end{cases}
\end{equation}
Collecting $\hat{\sigma}_{kl}$ for all pairs forms the empirical overlap matrix $\mathbf{\sigma} \in \mathbb{C}^{K \times K}$.
\newpara
The ratio form in \eqref{eqapp:estimatorsigma} is unbiased for the \emph{unnormalised} quantity $\Sigma_{kl}/N_k$, \emph{not} for the normalised overlap $\sigma_{kl}=\Sigma_{kl}/\sqrt{N_kN_l}$. In particular,
\begin{equation}
    \mathbb{E}_{q_k}\left[\expfunc{L_{kl}(x)}\right] = \frac{\Sigma_{kl}}{N_k} = \sigma_{kl} \sqrt{\frac{N_l}{N_k}},
\end{equation}
and the mixture variant to approximate the same expectation under $q_{\mathrm{mix}}$ inherits this scaling. This surrogate is therefore \emph{asymmetric} in $(k,l)$ and \emph{scale sensitive} to the norms $N_k$ and indeed is a \emph{biased} estimator of the normalised overlaps. In practice, two design choices made this surrogate behave well in our simulations: (i) we anneal $\lambda$ gently such that heads approach the degenerate manifold before strongly imposing orthogonality pressure, and (ii) the per-sample coefficients used in the penalty include a \emph{normalisation-correction} (see section \ref{app:subsecGradients} below) on the sampling head that discourages trivial rescaling. For all reported results, we validate by recomputing exact, normalised overlaps \emph{post-training} (via full enumeration for the small systems), which removes any residual ambiguity due to scaling. The post-training results indeed support the results obtained during the training.
\newpara
We optimise the weighted energy plus an orthogonal penalty cost
\begin{equation}
    C(\Theta) = \sum_{k = 1}^K w_k \hat{E}_k(\Theta) + \lambda \mathcal{P}(\hat{\sigma}(\Theta)),
\end{equation}
where $\sum_{k}w_k = 1$. In practice, we found the Frobenius penalty 
\begin{equation}
\label{eq:newpenaltoffdiag}
    \mathcal{P}(\sigma) = \lVert \sigma - \mathbf{1}_K\rVert^2_F = \sum_k |\sigma_{kk} - 1|^2 + \sum_{k \neq l}|\sigma_{kl}|^2.
\end{equation} 
convenient as it treats diagonal corrections and off-diagonals uniformly and mirrors diagnostics used post-optimisation. A linear anneal of the penalty strength $\lambda$ from a small value to its final value helps heads approach the degenerate manifold before enforcing strict orthogonality. 
\newpara
Note that in Appendix \ref{app:VMCNQSClasses}, an equivalent off-diagonal penalty $\frac{1}{2}\sum_{k \neq l} |\sigma_{kl}|^2$ is used, and switching between that and equation \eqref{eq:newpenaltoffdiag} only modifies small diagonal terms and corresponding gradients in a straightforward manner.


\subsection{\label{app:subsecGradients}Gradients}

The standard VMC identities from Appendix \ref{app:VMCNQSClasses} yield the standard gradients shown in equations \eqref{eqapp:energyderhead} and \eqref{eqapp:energydertrunk} whereby now, the expectations are taken with respect to $q_k$ (per-head) or as SNIS reweighted expectations under $q_{\mathrm{mix}}$. 
\newpara
For the overlap penalty \eqref{eq:newpenaltoffdiag}, differentiate through the estimator $\hat{\sigma}_{kl}$ in \eqref{eqapp:estimatorsigma}. Let
\begin{equation}
    r_{kl}(x) := e^{L_{kl(x)}} \implies \hat{\sigma}_{kl} = \sum_{i} \tilde{\omega}_{k, i} r_{kl}(x_i),
\end{equation}
where 
\begin{equation}
    \tilde{\omega}_{k,i} = \begin{cases}
        \frac{1}{S} &\farcomma x_i \sim q_k, \\
        \omega_{k, i} &\farcomma x_i \sim q_{\mathrm{mix}}.
    \end{cases}
\end{equation}
Now define the convenient product
\begin{equation}
    W_{kl, i} := \hat{\sigma}_{kl} r_{kl}(x_i).
\end{equation}
Then the stochastic gradient contributions that match the Appendix \ref{app:VMCNQSClasses} for the surrogate \eqref{eqapp:estimatorsigma} form up-to the ratio-estimator scaling under importance sampling can be expressed via per-sample coefficients and accumulated across samples/heads. 
\newpara
Specifically, because heads share the trunk $f_\vartheta$, one can assemble the full stochastic gradient from a single reverse-mode sweep through the trunk by forming a scalar function whose gradient equals the gradient of the cost. Concretely, let
\begin{align}
    F(\Theta) &= \sum_{k, i} \big[ c_{k, i, \mathrm{re}}^{(E)} \Re[\ell_k(x_i)] 
                + c_{k,i,\mathrm{im}}^{(E)} \Im[\ell_k(x_i)] \big] \nonumber\\
              &\quad + \sum_{k \neq l} \sum_i \big[ c_{kl, i, \mathrm{re}}^{(\mathcal{P})} \Re[\ell_k(x_i)] 
                + c_{kl, i, \mathrm{im}}^{(\mathcal{P})} \Im[\ell_k(x_i)] \big],
\end{align}
such that $\nabla_\Theta F (\Theta) = \nabla_\Theta C(\Theta)$. This implies that with the estimator $E_k$ of the form \eqref{eqapp:energyestimator},  and weights $\tilde{\omega}_{k,i} \in \lbrace 1/S, \omega_{k, i}\rbrace$, then
\begin{align}
    c_{k, i, \mathrm{re}}^{(E)} & = 2 \Re [E_{loc}^{(k)}(x_i) - \hat{E}_k] \tilde{\omega}_{k,i},
    \\
    c_{k, i, \mathrm{im}}^{(E)} & = 2 \Im [E_{loc}^{(k)}(x_i) - \hat{E}_k] \tilde{\omega}_{k,i}.
\end{align}
For the penalty term, a choice of per-sample coefficients that matches the unbiased gradient of \eqref{eq:newpenaltoffdiag} under the estimator \eqref{eqapp:estimatorsigma} is
\begin{align}
    c_{kl, i, \mathrm{re}}^{(\mathcal{P})} & = \lambda \tilde{\omega}_{k, i} \Re [W_{kl, i}], \\
    c_{kl, i, \mathrm{im}}^{(\mathcal{P})} & = - \lambda \tilde{\omega}_{k, i} \Im [W_{kl, i}],
\end{align}
for head $l$ and 
\begin{align}
    c_{kk, i, \mathrm{re}}^{(\mathcal{P})} & = \lambda \tilde{\omega}_{k, i} (2\Re[W_{kl, i}] - 2|\hat{\sigma}_{kl}|^2),
    \\
    c_{kk, i, \mathrm{im}}^{(\mathcal{P})} & = \lambda \tilde{\omega}_{k, i} \Im[W_{kl, i}],
\end{align}
for head k, for all $k \neq l$. The diagonal terms from $\lVert \sigma - \mathbf{1}_K\rVert^2_F$ are handled analogously.
\newpara
We note that our use of a \emph{biased} estimator of the normalised overlaps, as shown in \eqref{eqapp:estimatorsigma}, dictates a modification of the single-pass function $F$. Namely, differentiating the normalisation in $\sigma_{kl}$, the $\sqrt{N_k}$ factor contributes a constant centring term for head $k$ which shows up in the per-sampling coefficient as $-2|\hat{\sigma}_{kl}|^2$. This counters the tendency to reduce $r_{kl}$ by merely inflating $N_k$. Empirically, this steers optimisation toward reducing the numerator $\Sigma_{kl}$ of the normalised overlap (hence genuine orthogonality), not just reweighing norms. Finally, we note that an \emph{unbiased} estimator for such normalised overlaps can be in principle constructed, without the general scheme of the training protocol mentioned here being affected. 
\newpara
Given the above, a single call to reverse-mode automatic differentiation accumulates contributions from all heads through the shared trunk. In contrast, an MT-MH ensemble requires $K$ separate backward passes (one per trunk).


\subsection{Orthogonality residuals}

One diagnostic we use during and after training is the Frobenius norm of the head-overlap matrix. In what follows, we heuristically quantify an acceptable range for residual fluctuations in the Frobenius norm during training. 
\newpara
During training, two effects dominate residual fluctuations. Namely:
\begin{itemize}
    \item Monte Carlo noise: each entry $\hat{\sigma}_{kl}$ is a sample mean (per-head) or a self-normalised weighted mean (mixture). For $x_i \sim q_k$ and $S$ effective samples, then
    \begin{equation}
        \mathrm{Var}[\hat{\sigma}_{kl}] \approx \frac{v_{kl}}{S} \farcomma v_{kl} = \mathrm{Var}_{q_k}[r_{kl}(x)].
    \end{equation}
    Under the mixture approach with normalised weights, let $s_{kl}^2 = \sum_i \omega_{k, i}|r_{kl}(x_i) - \hat{\sigma}_{kl}|^2$. A large-sample approximation then gives
    \begin{equation}
    \mathrm{Var}[\hat{\sigma}_{kl}] \approx \frac{s^2_{kl}}{\mathrm{ESS}_k}.
    \end{equation}
    Aggregating entries (treating them as approximately independent at large batch size), the expected squared Frobenius error admits the crude estimate
    \begin{equation}
    \label{eqapp:essband}
        \mathbb{E}[\lVert \sigma - \mathbf{1}_K \rVert^2_F] \approx \sum_k \frac{s^2_{kk}}{\mathrm{ESS}_k} + \sum_{k \neq l} \frac{s^2_{kl}}{\mathrm{ESS}_k}.
    \end{equation}
    This motivates a simulation-specific tolerance band
    \begin{equation}
        \tau := \sqrt{\sum_{k \neq l} \frac{s^2_{kl}}{\mathrm{ESS}_k}}.
    \end{equation}
    That is, one declares orthogonality within Monte Carlo error if $\lvert\lvert \sigma - \mathbf{1}_K \rvert \rvert^2_F \lesssim 2\tau$.

    \item Stochastic-approximation drift: because $\lambda$ is annealed and gradients are noisy, the target $\sigma(\Theta)$ moves over iterations. Brief rebounds of the Frobenius norm were observed in optimisation to commonly coincide with dips in the $\mathrm{ESS}_k$, head trajectories crossing nodal surfaces, or when a stronger $\lambda$ trades a tiny energy gain for a modest overlap increase. These effects are transient and predicted by the ESS-based band in \eqref{eqapp:essband} above.
\end{itemize}


\section{\label{app:syssetup}Simulation parameters for full space resolution tests}

\begin{table*}[h]
\centering
\small
\setlength{\tabcolsep}{4pt}
\caption{\label{tab:syssetup}The simulation parameters used during all the simulations. Here, $N$ is the number of sites in the chain, $\eta$ is the learning rate, $n$ is the number of optimisation steps, $N_{MC}$ is the number of samples used in the Metropolis-Hastings sampler, $N_C$ is the number of chains, $\lambda_s$ is the starting penalty strength factor, $\lambda_f$ is the final penalty factor after annealing, $n_\lambda$ is the number of steps over which the penalty strength is annealed and $h$ is the trunk width.}
\begin{tabular}{ccccccccc}
\hline
$N$ & $\eta$ & $n$ & $N_{MC}$ & $N_C$ & $\lambda_s$ & $\lambda_f$ & $n_\lambda$ & $h$  
\\
\hline
4 & $1\times 10^{-3}$ & 1000 & 512 & 8 & $1\times 10^{-3}$ & $0.5$ & 200 & 32 \\
4(B) & $1\times 10^{-3}$ & 3000 & 1024 & 8 & $1\times 10^{-11}$ & $0.5$ & 1600 & 2 \\
4(C) & $1\times 10^{-3}$ & 2500 & 1024 & 8 & $1\times 10^{-8}$ & $1.5$ & 1200 & 4 \\
6 & $1\times 10^{-3}$ & 1000 & 512 & 8 & $1\times 10^{-3}$ & $0.5$ & 200 & 64 \\
8 & $1\times 10^{-3}$ & 1000 & 512 & 8 & $1\times 10^{-3}$ & $0.5$ & 200 & 64 \\
\hline
\end{tabular}
\end{table*}

\end{document}